\documentclass[reqno,11pt]{amsart}
\usepackage{multind}
\usepackage{graphicx}
\usepackage{amssymb}
\usepackage[mathscr]{eucal}
\numberwithin{equation}{section}
\allowdisplaybreaks[1]

\newcommand{\SetFigFont}[3]{}

\title[On the Support of Minimizers of Causal Variational Principles]
{On the Support of Minimizers of \\ Causal Variational Principles}

\author[F.\ Finster]{Felix Finster}
\thanks{Supported in part by the Deutsche Forschungsgemeinschaft.}

\author[D.\ Schiefeneder]{Daniela Schiefeneder \\ \\ December 2010 / April 2013}
\address{Fakult\"at f\"ur Mathematik \\ Universit\"at Regensburg \\ D-93040 Regensburg \\ Germany}
\email{finster@ur.de, Daniela.Schiefeneder@mathematik.uni-regensburg.de}
\newtheorem{Def}{Definition}[section]
\newtheorem{Thm}[Def]{Theorem}
\newtheorem{Prp}[Def]{Proposition}
\newtheorem{Lemma}[Def]{Lemma}

\newtheorem{Corollary}[Def]{Corollary}
\newtheorem{Example}[Def]{Example}
\newtheorem{Conjecture}[Def]{Conjecture}
\newcommand{\Proof}{\begin{proof}}
\newcommand{\QED}{\end{proof} \noindent}
\newcommand{\QEDrem}{\ \hfill $\Diamond$}

\newcommand{\bra}{\mbox{$< \!\!$ \nolinebreak}}
\newcommand{\ket}{\mbox{\nolinebreak $>$}}
\newcommand{\lbra}{\langle}
\newcommand{\lket}{\rangle}
\newcommand{\C}{\mathbb{C}}
\newcommand{\R}{\mathbb{R}}
\newcommand{\1}{\mbox{\rm 1 \hspace{-1.05 em} 1}}
\newcommand{\Z}{\mathbb{Z}}

\newcommand{\N}{\mathbb{N}}

\newcommand{\Tr}{\mbox{Tr\/}}

\newcommand{\beq}{\begin{equation}}
\newcommand{\eeq}{\end{equation}}

\renewcommand{\Im}{\text{\rm{Im}}}
\newcommand{\F}{{\mathscr{F}}}
\newcommand{\G}{{\mathscr{G}}}
\renewcommand{\L}{{\mathcal{L}}}
\newcommand{\Sact}{{\mathcal{S}}}
\newcommand{\Tact}{{\mathcal{T}}}
\newcommand{\Hil}{{\mathscr{H}}}
\newcommand{\D}{{\mathscr{D}}}
\newcommand{\I}{{\mathcal{I}}}
\newcommand{\J}{{\mathcal{J}}}
\newcommand{\K}{{\mathcal{K}}}
\newcommand{\Imc}{{\mathrm{Im}}\:}
\newcommand{\Rec}{{\mathrm{Re}}\:}
\newcommand{\Ff}{{\mathcal{F}}}
\newcommand{\Ll}{{\mathcal{L}}}

\newcommand{\Ss}{{\mathcal{S}}}
\newcommand{\Smin}{{\mathcal{S}_{\min}}}
\newcommand{\M}{{\mathfrak{M}}}

\newcommand{\Oo}{{\mathcal{O}}}
\newcommand{\bpm}{\begin{pmatrix}}
\newcommand{\epm}{\end{pmatrix}}
\newcommand{\vol}{{\mathrm{vol}}}
\newcommand{\dist}{{\mathrm{dist}}}
\newcommand{\dd}{{\mathsf{d}}}
\newcommand{\Thanks}{\vspace*{.5em} \noindent \thanks}
\newcommand{\la}{\langle}
\newcommand{\ra}{\rangle}
\newcommand{\Lin}{\text{\rm{L}}}

\DeclareMathOperator{\supp}{supp}

\newcommand{\U}{\text{\rm{U}}}
\newcommand{\vtmax}{\vartheta_{\max}}

\begin{document}
\maketitle

\begin{abstract}
A class of causal variational principles on a compact manifold is introduced and
analyzed both numerically and analytically.
It is proved under general assumptions that the support of a minimizing measure is
either completely timelike, or it is singular in the sense that its interior is empty.
In the examples of the circle, the sphere and certain flag manifolds,
the general results are supplemented by a more detailed and explicit analysis of the
minimizers. On the sphere, we get a connection to packing problems and the Tammes distribution.
Moreover, the minimal action is estimated from above and below.
\end{abstract}

\tableofcontents

\section{Introduction} \label{secintro}
Causal variational principles were proposed in~\cite{PFP} as an approach for formulating
relativistic quantum field theory (for surveys see~\cite{lrev, srev}). More recently, they were
introduced in a broader mathematical context as a class of nonlinear variational principles
defined on measure spaces~\cite{continuum}. Except for the examples and general
existence results in~\cite{discrete, small, continuum} and the symmetry breaking effect in the
discrete setting~\cite{osymm}, almost nothing is known on the structure of the minimizers.

The present paper is the first work dedicated to a detailed analysis of the minimizing measures
of causal variational principles. To introduce the problem, we begin with a
brief physical motivation (Section~\ref{secphys}) and then set up our mathematical framework
(Section~\ref{secmath}). Section~\ref{secresults} gives an overview of the obtained results.
In order to make the paper easily accessible also to mathematicians
who are not familiar with relativity or quantum field theory, in Section~\ref{secphys} we provide a
physical introduction that is independent of the other parts of the paper and may be skipped.

\subsection{Physical Background and Motivation} \label{secphys}
A relativistic quantum-mechanical particle is described by a Dirac wave function~$\psi$.
In Minkowski space~$M$, $\psi$ is a four-component complex-valued wave function which depends
on the spatial coordinates~$\vec{x}$ and time~$t$.
The absolute square~$|\psi(t, \vec{x})|^2$ of a Dirac wave function has the interpretation as the
probability density for the quantum mechanical particle to be at the position~$\vec{x}$.
Integrating this probability density over space and polarizing gives rise to a scalar product,
\[ \la \psi | \phi \ra := \int_{\R^3} \psi(t, \vec{x})^\dagger \phi(t, \vec{x})\: d^3 \vec{x} \]
(where~$\dagger$ denotes complex conjugation and transposition).
Note that the quantity $|\psi(t, \vec{x})|^2$ is not a scalar, but a density.
In order to obtain a Lorentz invariant scalar at a given space-time point~$x=(t, \vec{x})$, 
one must take the inner product~$\overline{\psi(x)} \phi(x)$, where~$\overline{\psi} \equiv \psi^\dagger
\gamma^0$ is the so-called adjoint spinor. This inner product is indefinite of signature~$(2,2)$.
Likewise, in the presence of a gravitational field, space-time is described by a Lorentzian
manifold~$(M, g)$. In this case, the Dirac wave functions are sections of the so-called spinor bundle.
Again, the spatial integral of the probability density gives rise to a scalar product~$\la.|.\ra$. Moreover, 
the spinors at any space-time point~$x \in M$ are again endowed with an inner
product~$\overline{\psi(x)} \phi(x)$ of signature~$(2,2)$.

Let us consider a system consisting of several Dirac particles in a concrete physical configuration
(say an atom or molecule). For simplicity, we here describe the system by the corresponding
one-particle wave functions~$\psi_1, \ldots, \psi_f$ (for the connection to Fock spaces and the Pauli
exclusion principle see for example~\cite{entangle}). Moreover,
in order to keep the presentation as simple as possible, we assume that the number~$f$
of particles is finite and that the wave functions~$\psi_i$ are continuous. In order to work in a basis
independent way, we consider the vector space~$\Hil := \bra \psi_1, \ldots, \psi_f \ket$
spanned by the one-particle wave functions. Endowed with the above scalar product~$\la .|. \ra$,
this vector space is an $f$-dimensional Hilbert space~$(\Hil, \la .|. \ra_\Hil)$.
For any space-time point~$x$, we introduce the {\em{local correlation operator}}~$F(x)$ as the
signature operator of the inner product at~$x$ expressed in terms of the Hilbert space scalar product,
\[ \overline{\psi(x)} \phi(x) = -\la \psi | F(x) \phi \ra_\Hil \qquad \text{for all~$\psi, \phi \in \Hil$}\:. \]
Taking into account that the inner product at~$x$ has signature~$(2,2)$,
the local correlation operator is a symmetric operator in~$\Lin(\Hil)$
of rank at most four, which has at most two positive and at most two negative eigenvalues.

The local correlation operators encode how the Dirac particles are distributed in space-time
and how the wave functions are correlated at the individual space-time points.
We want to consider the local correlation operators as the basic objects in space-time,
meaning that all geometric and analytic structures of space-time (such as the causal
and metric structure, connection, curvature, gauge fields, etc.) should be deduced
from the operators~$F(x)$. The only structure besides the local correlation operators
which we want to keep is the volume measure of space-time.
In order to implement this concept mathematically, it is useful to
identify a space-time point~$x$ with the corresponding local correlation operator~$F(x) \in \Lin(\Hil)$.
Then space-time is identified with the subset~$F(M) \subset \Lin(\Hil)$.
Next, we introduce the {\em{universal measure}}~$\rho= F_* \mu_M$
as the push-forward of the volume measure on~$M$ under the mapping~$F$
(thus~$\rho(\Omega) := \mu_M(F^{-1}(\Omega))$, where~$d\mu_M
= \sqrt{|\deg g|}\, d^4x$ is the standard volume measure).
Dropping all the additional structures of space-time, we 
obtain a causal fermion system of spin dimension two as defined in~\cite[Section~1.2]{rrev}:

\begin{Def} {\em{
Given a complex Hilbert space~$(\Hil, \la .|. \ra_\Hil)$ (the {\em{``particle space''}})
and a parameter~$n \in \N$ (the {\em{``spin dimension''}}), we let~$\F \subset \Lin(\Hil)$ be the set of all
self-adjoint operators on~$\Hil$ of finite rank, which (counting with multiplicities) have
at most~$n$ positive and at most~$n$ negative eigenvalues. On~$\F$ we are given
a positive measure~$\rho$ (defined on a $\sigma$-algebra of subsets of~$\F$), the so-called
{\em{universal measure}}. We refer to~$(\Hil, \F, \rho)$ as a {\em{causal fermion system in the
particle representation}}.
}}
\end{Def} \noindent

Causal fermion systems provide a general abstract mathematical framework for the formulation
of relativistic physical theories.
Clearly, by choosing~$\Hil$ as an infinite-dimensional Hilbert space, one can
describe an infinite number of particles. The setting is so general that it allows for the
description of continuum space-times (such as Minkowski space or a Lorentzian manifold),
discrete space-times (such as a space-time lattice) and even so-called ``quantum space-times''
which have no simple classical correspondence (for examples, we refer the reader to~\cite{small, lqg}
and~\cite{finite}). It is a remarkable fact that in causal fermion systems there are
many inherent geometric and analytic structures. Namely, starting from a general
causal fermion system, one can deduce space-time together with geometric structures
which generalize the setting of spin geometry (see~\cite{lqg}).
Moreover, the causal action principle gives interesting analytic structures (see~\cite{continuum}).
In particular, one obtains a background-free formulation of quantum
field theory in which ultraviolet divergences of standard quantum field theory are avoided.
Since the setting also works in curved space-time, causal fermion systems are also a promising
approach for quantum gravity. We refer the interested reader to the review article~\cite{lqg}
as well as to the preprints~\cite{sector, lepton, quark}.

We now outline how to deduce the geometric and analytic structures which will
be relevant in this paper. Starting from a causal fermion system~$(\Hil, \F, \rho)$, one
defines space-time as the support of the universal measure, $M := \supp \rho$.
On~$M$, we consider the topology induced by~$\F \subset \Lin(\Hil)$.
The causal structure is encoded in the spectrum of the operator products~$x y$:
\begin{Def} \label{def2}  {\em{ For any~$x, y \in \F$, the product~$x y$ is an operator
of rank at most~$2n$. We denote its non-trivial eigenvalues
by $\lambda^{xy}_1, \ldots, \lambda^{xy}_{2n}$ (where we count with algebraic multiplicities).
The points~$x$ and~$y$ are called {\em{timelike}} separated if the~$\lambda^{xy}_j$ are all real.
They are said to be {\em{spacelike}} separated if all the~$\lambda^{xy}_j$ are complex
and have the same absolute value.
In all other cases, the points~$x$ and~$y$ are said to be {\em{lightlike}} separated.
}}
\end{Def}
The main analytic tool is the causal variational principle defined as follows.
For two points~$x, y \in \F$ we define the {\em{spectral weight}} $|.|$
of the operator products~$xy$ and~$(xy)^2$ by
\[ |xy| = \sum_{i=1}^{2n} |\lambda^{xy}_i|
\qquad \text{and} \qquad \left| (xy)^2 \right| = \sum_{i=1}^{2n} | \lambda^{xy}_i |^2 \:. \]
We also introduce the
\beq \label{Lagrange}
\text{\em{Lagrangian}} \qquad \L(x,y) = \big| (xy)^2 \big| - \frac{1}{2n}\: |xy|^2 \:.
\eeq
For a given universal measure~$\rho$ on~$\F$, we define the non-negative functionals
\begin{align}
\text{\em{action}} \qquad \Sact[\rho] &= \iint_{\F \times \F} \L(x,y)\: d\rho(x)\, d\rho(y)  \\
\text{\em{constraint}} \qquad \Tact[\rho] &= \iint_{\F \times \F} |xy|^2\: d\rho(x)\, d\rho(y) \label{Tdef}\:.
\end{align}
The {\em{causal action principle}} is to
\[ \text{minimize~$\Sact$ for fixed~$\Tact$} \]
under variations of the universal measure.
These variations should keep the total volume unchanged, which means that a
variation~$(\rho(\tau))_{\tau \in (-\varepsilon, \varepsilon)}$
should satisfy the conditions
\[ \big| \rho(\tau) - \rho(\tau') \big|(\F) < \infty \qquad \text{and} \qquad
\big( \rho(\tau) - \rho(\tau') \big) (\F) = 0 \]
for all~$\tau, \tau' \in (-\varepsilon, \varepsilon)$ 
(where~$|.|$ denotes the total variation of a measure; see~\cite[\S28]{halmosmt}).
Depending on the application, one may impose one of the following constraints:
\begin{itemize}
\item[(C1)] The {\bf{trace constraint}}:
\[ \int_\F \Tr (x) \, d\rho(x) = f\:. \]
\item[(C2)] The {\bf{identity constraint}}:
\[ \int_\F x \, d\rho(x) = \1_\Hil\:. \]
\item[(C3)] {\bf{Prescribing $2n$ eigenvalues}}: We denote the non-trivial eigenvalues of~$x$ counted with
multiplicities by~$\nu_1, \ldots, \nu_{2n}$ and order them such that
\[ \nu_1 \leq \cdots \leq \nu_{n} \;\leq 0 \;\leq\; \nu_{n+1} \leq \ldots \leq \nu_{2n} \:. \]
For given constants~$c_1,\ldots, c_{2n}$, we impose that
\beq \label{pin}
\nu_j(x) = c_j \qquad \text{for all } x \in M \text{ and } j=1,\ldots, 2n\:.
\eeq
\end{itemize}
Moreover, one may prescribe properties of the universal measure
by choosing a measure space~$(\hat{M}, \hat{\mu})$ and restricting attention to universal measures
which can be represented as the push-forward of~$\hat{\mu}$,
\beq \label{pushforward}
\rho = F_* \hat{\mu} \qquad \text{with} \qquad \text{$F \::\: \hat{M} \rightarrow \F$ measurable}\:.
\eeq
One then minimizes the action under variations of the mapping~$F$.

The Lagrangian~\eqref{Lagrange} is compatible with our notion of causality in the
following sense. Suppose that two points~$x, y \in \F$ are spacelike separated
(see Definition~\ref{def2}). Then the eigenvalues~$\lambda^{xy}_i$ all have the same absolute value,
so that the Lagrangian~\eqref{Lagrange} vanishes. Thus pairs of points with spacelike separation do not
enter the action. This can be seen in analogy to the usual notion of causality where
points with spacelike separation cannot influence each other.

In the present paper, we restrict attention to a special class of causal variational principles,
as we now explain. First, we always assume that the number~$f$ of particles is finite.
Then~$\Hil$ is finite-dimensional, and~$\F$ is a locally compact topological space.
Moreover, we always have the situation in mind when the above constraint~(C3) is present.
In this case, we may replace~$\F$ by the set of those operators which satisfy
the constraints~\eqref{pin}. Then $\F \subset \Lin(\Hil)$ becomes a compact topological space.
In this so-called {\em{compact setting}}, the causal variational principle is mathematically well-defined,
even without imposing the constraint~\eqref{Tdef}. For simplicity, we shall disregard this constraint.
Next, we do not want to use~\eqref{pushforward} to prescribe any properties of~$\rho$.
Instead, we allow~$\rho$ to be any normalized positive Borel measure on~$\F$.
This is what we mean by the {\em{continuous setting}}. Finally, we assume that the
Lagrangian~\eqref{Lagrange} is of the form
\beq \label{Ljusty}
\L = \max(0, \D) \qquad \text{with} \qquad \D \in C^\infty(\F \times \F, \R) \:.
\eeq
This assumption is automatically satisfied in the case of spin dimension one. Namely, in this case the
eigenvalues~$\lambda^{xy}_1$ and~$\lambda^{xy}_2$ are either both real, or else they form a
complex conjugate pair. In the first case, the Lagrangian is positive and smooth, whereas in the second
case it vanishes identically. Hence~$\L$ can indeed be written in the desired form~\eqref{Ljusty}.

The above specializations are made in order to make the mathematical problem manageable.
Despite the major simplifications, our results shed some light on the general
structure of minimizers of causal variational principles.
Qualitatively speaking, our results reveal a mechanism which favors
discrete over continuous configurations. This mechanism seems to have interesting physical implications:
First, thinking of a quantum space-time, the discreteness of the minimizing measure means that
the causal action principle should arrange that space-time is discrete on the Planck scale.
This effect would be very desirable because it would resolve the ultraviolet problems
of quantum field theory.
Another implication is related to field quantization.
To explain the connection, let us consider a family of universal measures~$(\rho_\tau)_{\tau \in \R}$
which describes a physical system where a classical amplitude is varied continuously.
Then our ``discreteness results'' mean that some discrete values of~$\tau$ should be favored.
As a consequence, the amplitude no longer takes continuous values, but it should be ``quantized''
to take the discrete values for which the causal action has a local minimum.
Finally, our ``discreteness results'' might also account for effects related to the wave-particle duality
and the collapse of the wave function in the measurement process.
For a more detailed physical discussion we refer to~\cite{rev, dice2010}.

\subsection{The Mathematical Setup} \label{secmath}
We now introduce our mathematical setting.
Let~$\F$ be a smooth compact manifold (of arbitrary dimension). For a given function
\begin{align} 
\D \in C^\infty(\F \times \F, \R) \quad &\text{being symmetric: }
\D(x,y) = \D(y,x) \;\; \forall \:x,y \in \F \label{Ddef} \\
&\text{and strictly positive on the diagonal: $\D(x,x)>0\:,$} \label{posdiagonal}
\end{align}
we define the {\em{Lagrangian}}~$\L$ by
\beq \label{Lform}
\L = \max(0, \D) \in C^{0,1}(\F \times \F, \R^+_0) \:.
\eeq
Introducing the {\em{action}}~$\Sact$ by
\beq \label{Sdef}
\Sact[\rho] = \iint_{\F \times \F} \L(x,y)\: d\rho(x)\: d\rho(y)\:,
\eeq
our action principle is to
\beq \label{var}
\text{minimize~$\Sact$ under variations of~$\rho \in \M\:,$}
\eeq
where~$\M$ denotes the set of all normalized positive regular Borel measures on~$\F$.
In view of the symmetric form of~\eqref{Sdef}, it is no loss of generality to assume that~$\L(x,y)$
is symmetric in~$x$ and~$y$. Therefore, it is natural to assume that also~$\D(x,y)$ is symmetric~\eqref{Ddef}. If~\eqref{posdiagonal} were violated, every measure supported in the
set~$\{x \,|\, \D(x,x) \leq 0\}$ would be a minimizer.
Thus the condition~\eqref{posdiagonal} rules out trivial cases.

The existence of minimizers follows immediately from abstract compactness arguments
(see~\cite[Section~1.2]{continuum}).
\begin{Thm}
The infimum of the variational principle~\eqref{var} is attained in~$\M$.
\end{Thm} \noindent
Note that the minimizers will in general not be unique. Moreover, the abstract framework gives
no information on what the minimizers look like.

The notion of causality can now be introduced via the sign of~$\D$.
\begin{Def}[{\bf{causal structure}}]
\[ \text{Two points~$x, y \in \F$ are called} \; \left\{ \begin{array}{c} \text{timelike}
\\ \text{lightlike} \\ \text{spacelike} \end{array} \right\} \; \text{separated if} \;
\left\{ \begin{array}{c} \D(x,y) > 0 \\ \D(x,y) = 0 \\ \D(x,y) < 0 \!\! \end{array} \right\} . \]
We define the sets
\begin{align*}
\I(x) &= \{ y \in \F \text{ with } \D(x,y) > 0 \} &\quad& \text{open lightcone} \\
\J(x) &= \{ y \in \F \text{ with } \D(x,y) \geq 0 \} && \text{closed lightcone} \\
\K(x) &= \partial \I(x) \cap \partial \big( \F \setminus \J(x) \big) && \text{boundary of the lightcone}\:.
\end{align*}
\end{Def} \noindent
Thus~$y \in \K(x)$ if and only if the function~$\D(x,.)$ changes sign in every neighborhood of~$y$.

Our action is compatible with the causal structure in the sense that if~$x$ and~$y$ have lightlike
or spacelike separation, then the Lagrangian vanishes, so that the pair~$(x,y)$ does not contribute to
the action. Note that for a given minimizer~$\rho$, we have similarly a causal structure
on its support.

In order to work in more specific examples, we shall consider the following three model problems.
\begin{itemize}
\item[(a)] {\em{Variational principles on the sphere:}} \label{exsphere} \\
We consider the setting of~\cite[Chapter~1]{continuum} in the case~$f=2$
(see also~\cite[Examples~1.5, 1.6 and~2.8]{continuum}).
Thus  for a given parameter~$\tau\geq1$, we let~$\F$ be the space of Hermitian $(2 \times 2)$-matrices
whose eigenvalues are equal to~$1+\tau$ and~$1-\tau$.
Writing a matrix~$F \in \F$ as a linear combination of Pauli matrices,
\[ F = \tau\: x \!\cdot\! \sigma + \1 \qquad \text{with} \qquad x \in S^2 \subset \R^3 \:, \]
we can describe~$F$ by the unit vector~$x$ (here~$\cdot$ denotes the scalar product in~$\R^3$).
Thus~$\F$ can be identified with the unit sphere~$S^2$.
The function~$\D$ is computed in~\cite[Example~2.8]{continuum} to be
\beq \label{DS2}
\D(x,y) = 
2 \tau^2\: (1+ \langle x,y \rangle) \left( 2 - \tau^2 \:(1 - \langle x,y \rangle) \right) .
\eeq
This function depends only on the angle~$\vartheta_{xy}$ between the points~$x, y \in S^2$
defined by~$\cos \vartheta_{xy} = \langle x, y \rangle$,
which clearly coincides with the geodesic distance of~$x$ and~$y$.
\begin{figure}
 \centering
 \includegraphics[width=9cm]{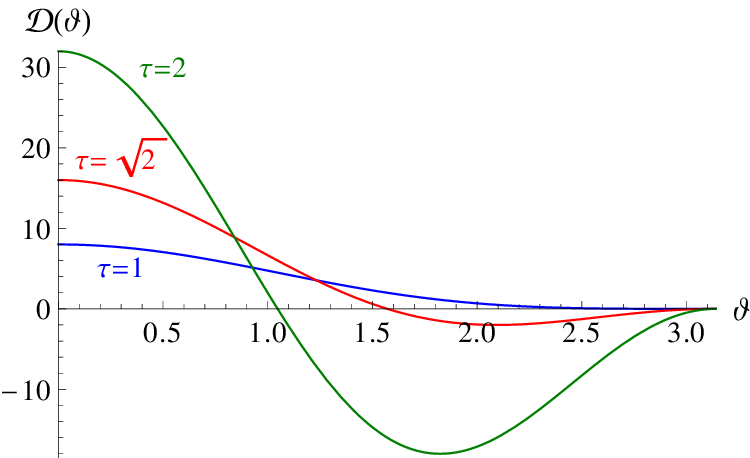}
 \caption{The function $\D$.}
 \label{plotD}
\end{figure}
Considered as function of~$\vartheta\in [0,\pi]$, $\D$ has its maximum at~$\vartheta=0$
and is minimal if~$\cos(\vartheta)=-\tau^{-2}$. Moreover, $\D(\pi)=0$.
Typical plots are shown in Figure~\ref{plotD}.
In the case~$\tau>1$, the function~$\D$ has two zeros at~$\pi$ and
\begin{equation}\label{thetamax}
\vartheta_{\max}:=\arccos\left(1-\frac{2}{\tau^2}\right) . \end{equation} 
In view of~\eqref{Lform}, the Lagrangian is positive if and only if~$0\leq \vartheta<\vartheta_{\max}$. Thus~$\I(x)$ is an open spherical cap, and $\J(x)$ is its closure together with the antipodal point of~$x$, 
\[ \I(x)= \Big\{y:\langle x,y\rangle > 1-\frac{2}{\tau^2} \Big\} \:,\qquad
\J(x)=\overline{\I(x)}\cup \{-x\} \:. \]
If~$\tau$ is increased, the opening angle~$\vartheta_{\max}$ of the lightcones gets smaller.
In the degenerate case~$\tau=1$, the function~$\D$ is decreasing, non-negative and has exactly one zero
at $\vartheta=\pi$. Hence the Lagrangian~$\Ll$ coincides with $\D$. All points on the sphere are
timelike separated except for the antipodal points.
The lightcones are given by~$\I(x)=S^2\backslash \{-x\}$ and~$\J(x)=S^2$.

If we regard~$\rho$ as a density on the sphere, the
action~\eqref{Sdef} looks like the energy functional corresponding to
a pair potential~$\Ll$ (see for example~\cite{saff+kuijlaars}).
Using physical notions, our pair potential is repelling (because~$\Ll(\vartheta)$ is a decreasing function)
and has short range (because~$\L$ vanishes if~$\vartheta \geq \vartheta_{\max}$).
\item[(b)] {\em{Variational principles on the circle:}} \\ \label{Bcircle}%
In order to simplify the previous example, we set~$\F = S^1$. For~$\D$ we
again choose~\eqref{DS2}.
\item[(c)] {\em{Variational principles on the flag manifold~$\F^{1,2}(\C^f)$:}} \\
As in~\cite[Chapter~1]{continuum}, for a given parameter~$\tau > 1$ and an integer parameter~$f>2$,
we let~$\F$ be the space of Hermitian $(f \times f)$-matrices of rank two, whose
nontrivial eigenvalues are equal to~$1+\tau$ and~$1-\tau$. Every~$x \in \F$ is uniquely described
by the corresponding eigenspaces~$U$ and~$V$.
By considering the chain~$U \subset (U \cup V)$, $x$ can be identified with an element of the flag manifold $\F^{1,2}(\C^f)$, the space of one-dimensional subspaces contained in a two-dimensional
subspace of $\C^f$ (see~\cite{helgason2}). It is a $(4f-6)$-dimensional compact manifold.
Every~$U \in \U(f)$ gives rise to the mapping~$x \rightarrow U x U^{-1}$ on~$\F$.
This resulting group action of~$\U(f)$ on~$\F$ acts transitively, making~$\F$
a homogeneous space (see~\cite{helgason2} for details).

For two points~$x, y \in \F$, we denote the two non-trivial eigenvalues of the matrix product~$x y$
by~$\lambda_+^{xy}, \lambda_-^{xy} \in \C$ and define the Lagrangian by
\[ \L(x,y) = \frac{1}{2} \left(|\lambda_+^{xy}| - |\lambda_-^{xy}| \right)^2 \:. \]
This Lagrangian is $\U(f)$-invariant. In order to bring it into a more convenient form,
we first note that by restriction to the image of~$y$, the characteristic polynomial
of~$xy$ changes only by irrelevant factors of~$\lambda$,
\[ \det(x y - \lambda \1) = \lambda^{f-2} \; \det \left( (\pi_y \,x y- \lambda \1)|_{\text{Im}\, y} \right) , \]
where~$\pi_y$ denotes the orthogonal projection to~$\text{Im} \,y$.
It follows that~$\lambda_+^{xy}$ and~$\lambda_-^{xy}$ are the eigenvalues of the
$(2 \times 2)$-matrix $\pi_y x y|_{\text{Im}\, y}$. In particular,
\[ \lambda_+^{xy} \lambda_-^{xy} = \det (\pi_y x y|_{\text{Im}\, y})
= \det (\pi_y x \pi_y|_{\text{Im}\, y}) \, \det(y|_{\text{Im}\, y}) \geq 0 \:, \]
because the operator~$\pi_y x \pi_y$ again has at most one positive and one negative eigenvalue.
Moreover, the relation~$\lambda_+^{xy} +\lambda_-^{xy}=\Tr(xy)\in \R$
shows that the two eigenvalues are either both real and have the same sign or else they
form a complex conjugate pair, in which case the Lagrangian vanishes. Finally, using
that~$\left(\lambda_+^{xy}\right)^2 +\left(\lambda_-^{xy}\right)^2=\Tr \big((x y)^2 \big)$, the Lagrangian
can be written in the form~\eqref{Lform} with
\beq \label{Dsimple}
\D(x,y) =  \frac{1}{2} \left(\lambda_+^{xy} - \lambda_-^{xy} \right)^2 =
\Tr \big((x y)^2 \big) - \frac{1}{2} \big( \Tr(xy) \big)^2 \:.
\eeq
\end{itemize}
We finally comment on the limitations of our setting and mention possible generalizations.
First, we point out that our structural results do not immediately apply in the cases when~$\F$ is
non-compact or when additional constraints are considered (see~\cite[Chapter~2]{continuum}).
However, it seems that in the non-compact case, our methods and results
could be adapted to the so-called moment measures as introduced in~\cite[Section~2.3]{continuum}.
A promising strategy for handling additional constraints would be to first derive the corresponding
Euler-Lagrange equations, treating the constraints with Lagrange multipliers
(for details see~\cite{lagrange}).
Then one could
try to recover these Euler-Lagrange equations as those corresponding to an unconstrained
variational problem on a submanifold~$\G \subset \F$, where our methods could again be used.
We finally point out that in the case of higher spin dimension~$n>1$, it is in general impossible
to write the Lagrangian in the form~\eqref{Lform} with a smooth function~$\D$, because
the Lagrangian is in general only Lip\-schitz continuous in the open light cone.
A possible strategy would be to first show that the support of~$\rho$ lies on
a submanifold~$\G \subset \F$, and then to verify that by restricting~$\L$ to~$\G \times \G$,
it becomes smooth in the open lightcones.

\subsection{Overview of the Main Results} \label{secresults}
We now outline our main results.
In Section~\ref{sec3}, we present numerical results on the sphere (see Figure~\ref{figvgl_weight})
and discuss all the main effects which will be treated analytically later on.
In Section~\ref{sec4}, we derive general results on the structure of the minimizers.
We first derive the corresponding Euler-Lagrange equations and conditions
for minimality (see Lemma~\ref{lemmaEL} and Lemma~\ref{lemmaP}).
We then prove under general assumptions that the minimizers are either generically 
timelike (see Definition~\ref{defgtl}) or else the support of the minimizing measure~$\rho$ defined by
\[ \supp \rho = \{ x \in \F \:|\:  \text{$\rho(U) \neq 0$ for every open neighborhood~$U$ of~$x$} \} \]
is singular in the sense that its interior is empty (see Theorems~\ref{thmmain1} and~\ref{thmmain2}).
In the following sections, we apply these general results to our model examples and derive more detailed
information on the minimizers. 
In Section~\ref{sec5}, we consider the variational principle on the circle.
After briefly discussing numerical results (see Figure~\ref{vgl_S1}),
we prove a ``phase transition'' between generically timelike minimizers and minimizers with
singular support and construct many minimizers in closed form (see Corollary~\ref{corS1} and
Theorem~\ref{thmS1}). In Section~\ref{sec6}, the variational principle on the sphere is considered.
We again prove the above phase transition (see Corollary~\ref{cor61}) and estimate the action
from above and below (see Figure~\ref{figestimates} and
Proposition~\ref{prp_k_gen}). Finally, in Section~\ref{sec7} we
apply our general results to flag manifolds (see Theorem~\ref{corfm}). Moreover, we prove
that minimizers with singular support exist (see Theorem~\ref{thmfm}) and
give an outlook on generically timelike minimizers.

We close the introduction with a remark on our methods. Generally speaking, our
proofs are based on a detailed functional analytic study of first and second variations.
The analysis of second variations reveals that, for a minimizing measure~$\rho$,
a certain operator~$\L_\rho$ on a Hilbert space~$(\Hil_\rho, \lbra .,. \lket_\rho)$
is positive (see Lemma~\ref{lemmaP}).
This positivity property is the key for analyzing the support of~$\rho$.
Our method is to proceed indirectly and to show that if the support does not have the
desired properties, then there is a vector~$\psi \in \Hil_\rho$ with~$\lbra \psi, \L_\rho \psi \lket_\rho < 0$,
a contradiction. Clearly, the difficulty is to construct the vector~$\psi$; here we use the specific form
of our variational principle.

The method of relating properties of the minimizer to positivity properties of a corresponding operator
bears some similarity to the methods in~\cite{suto, capet+friesecke}.
Namely, in~\cite{suto} it is shown that if a certain pair potential is positive in momentum space,
then there are discrete ground states. In~\cite{capet+friesecke}, on the other hand,
it is shown for a repulsive pair potential that in a suitable limit where the number of particles tends
to infinity, the particles are distributed uniformly on a $2$-sphere.
The methods in both of these papers cannot be compared directly, because the mathematical structure
of the considered variational problems is quite different from our causal variational principles.
In particular, our variational principles are not formulated for particle configurations,
but for general measures. Moreover, we make essential use of the fact that our pair potential is not
smooth, but only Lipschitz continuous.

\section{Numerical Results on the Sphere} \label{sec3}
In order to motivate our general structural results, we now describe our findings in a
numerical analysis of the variational principle on the sphere (see Example~(a)
on page~\pageref{exsphere}). Clearly, in a numerical study
one must work with discrete configurations. Our first attempt is to choose a finite number of
points~$x_1, \ldots, x_m \in S^2$ and to let~$\rho$ be the corresponding normalized
{\em{counting measure}}, i.e.
\beq \label{count}
\int_{S^2} f\, d\rho \::=\: \frac{1}{m} \sum_{i=1}^m f(x_i) \qquad \forall f \in C^0(S^2)\:.
\eeq
Then the action~\eqref{Sdef} becomes
\beq \label{Scount}
\Ss = \frac{1}{m^2} \sum_{i,j=1}^m \Ll(x_i, x_j)\:.
\eeq
By varying the points~$x_i$ for fixed~$m$, we obtain a minimizer~$\rho_m$.
Since every normalized positive regular Borel measure can be approximated by such counting
measures, we can expect that if we choose~$m$ sufficiently large, the measure~$\rho_m$
should be a good approximation of a minimizing measure~$\rho \in \M$
(more precisely, we even know that $\rho_m \rightarrow \rho$ as $m \rightarrow \infty$ with
convergence in the weak $(C^0)^*$-topology).

If~$\tau$ is sufficiently large, the opening angle of the lightcones is so small that
the~$m$ points can be distributed on the sphere such that any two different points are
spacelike separated. In this case, the action becomes
\beq \label{Tammes}
\Ss = \frac{1}{m}\: \Ll(\vartheta=0) \:,
\eeq
and in view of~\eqref{Scount} this is indeed minimal. The question for which~$\tau$
such a configuration exists leads us to the {\em{Tammes problem}}, a packing problem where
the points are distributed on
the sphere such that the minimal distance~$\vartheta_m$
between distinct points is maximized, see \cite{sloane}.
More precisely, we know that the Tammes distribution is a minimizer of our action if~$\tau$ is so large
that~$\vartheta_m > \vartheta_{\text{max}}$.
Until now, the Tammes problem is solved only if $m \leq 12$ and for~$m=24$ (for details see~\cite{croft}
and the references therein). For special values of~$m$, the solutions of the Tammes problem
are symmetric solids such as the tetrahedron ($m=4$), the octahedron ($m=6$),
the icosahedron ($m=12$) and the snub cube ($m=24$). Moreover, much research has been done on the numerical evaluation of spherical codes, mostly by N.J.A.\ Sloane, with the collaboration of R.H.\ Hardin,
W.D.\ Smith and others, \cite{sloane}, obtaining numerical solutions of the Tammes problem for up to $130$ points.

In the case~$\vartheta_m < \vartheta_{\max}$, the measure~$\rho_m$ was constructed
numerically using a simulated annealing algorithm\footnote{We use the
\emph{``general simulated annealing algorithm''} by J.\ Vandekerckhove, 
\copyright~2006, {\tt{http://www.mathworks.de/matlabcentral/fileexchange/10548}}.}.
In order to get optimal results, we used this algorithm iteratively, using either a Tammes distribution
or previous numerical distributions as starting values. Using that~$\D$ depends smoothly on~$\tau$,
it is useful to increase or decrease~$\tau$ in small steps, and to use the numerical minimizer as the
starting configuration of the next step. In Figure~\ref{figvgl_dis}, the numerically found~$\Sact[\rho_m]$
is plotted for different values of~$m$ as a function of the parameter~$\tau$.
 \begin{figure}
 \centering
 \includegraphics[width=8cm]{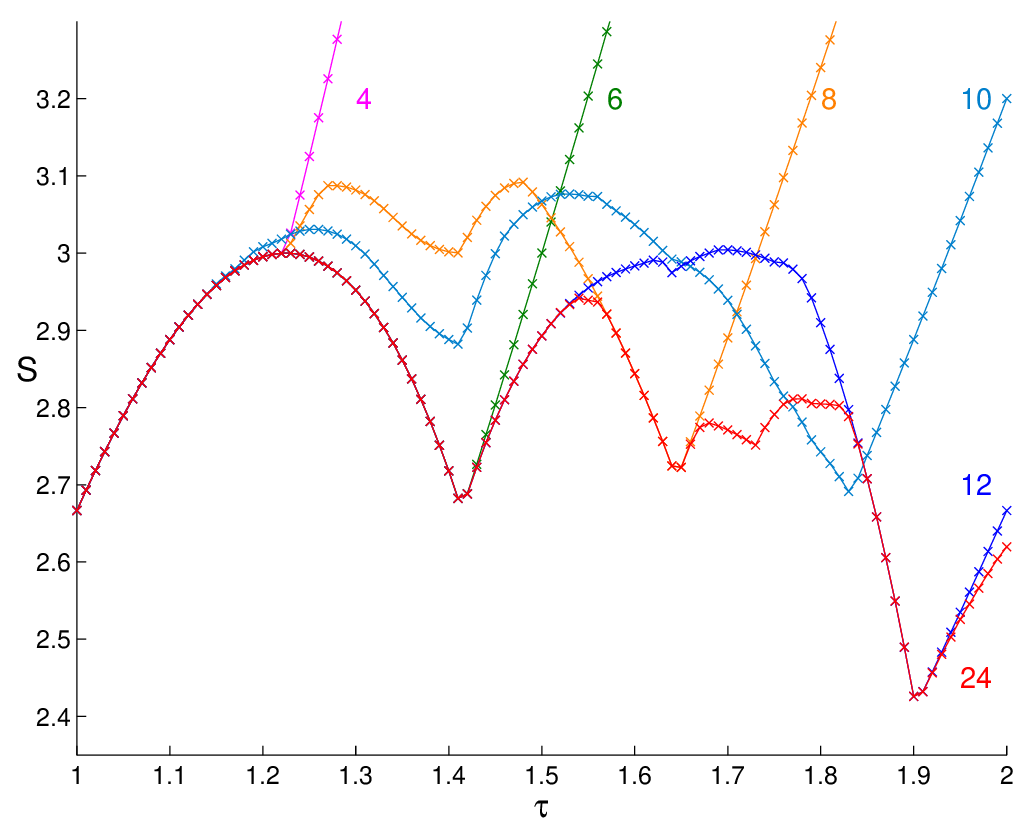}
 \caption{Numerical minima for the counting measure on the sphere.}
 \label{figvgl_dis}
 \end{figure}
The resulting plots look rather complicated. The considered values for~$m$ are too small
for extrapolating the limiting behavior as~$m \rightarrow \infty$. Nevertheless, one observation
turned out to be very helpful: Near~$\tau \approx 1.2$, the plots for different values of~$m$ look the same.
The reason is that some of the~$x_i$ coincide, forming ``clusters'' of several points.
For example, in the case~$m=12$, the support of $\rho$ consists of only six distinct points,
each occupied by two~$x_i$. A similar ``clustering effect'' also occurs for higher $\tau$ if~$m$
is sufficiently large.

These findings give the hope that for large~$m$, the minimizers might be well-approximated
by a measure supported at a few cluster points, with weights counting the number of points at each cluster.
This was our motivation for considering a {\em{weighted counting measure}}. Thus for
any fixed~$m$, we choose points~$x_1, \ldots, x_m \in S^2$ and
corresponding weights~$\rho_1, \ldots, \rho_m$ with
\[ \rho_i \geq 0 \qquad \text{and} \qquad \sum_{i=1}^m \rho_i = 1\:. \]
We introduce the corresponding measure~$\rho$ in generalization of~\eqref{count} by
\beq \label{weightcount}
\int_{S^2} f\, d\rho \::=\: \sum_{i=1}^m \rho_i \,f(x_i) \qquad \forall f \in C^0(S^2)\:.
\eeq
Seeking for numerical minimizers by varying both the points~$x_i$ and the weights~$\rho_i$,
we obtain the plots shown in Figure~\ref{figvgl_weight}.
 \begin{figure}
  \centering
 \includegraphics[width=10cm]{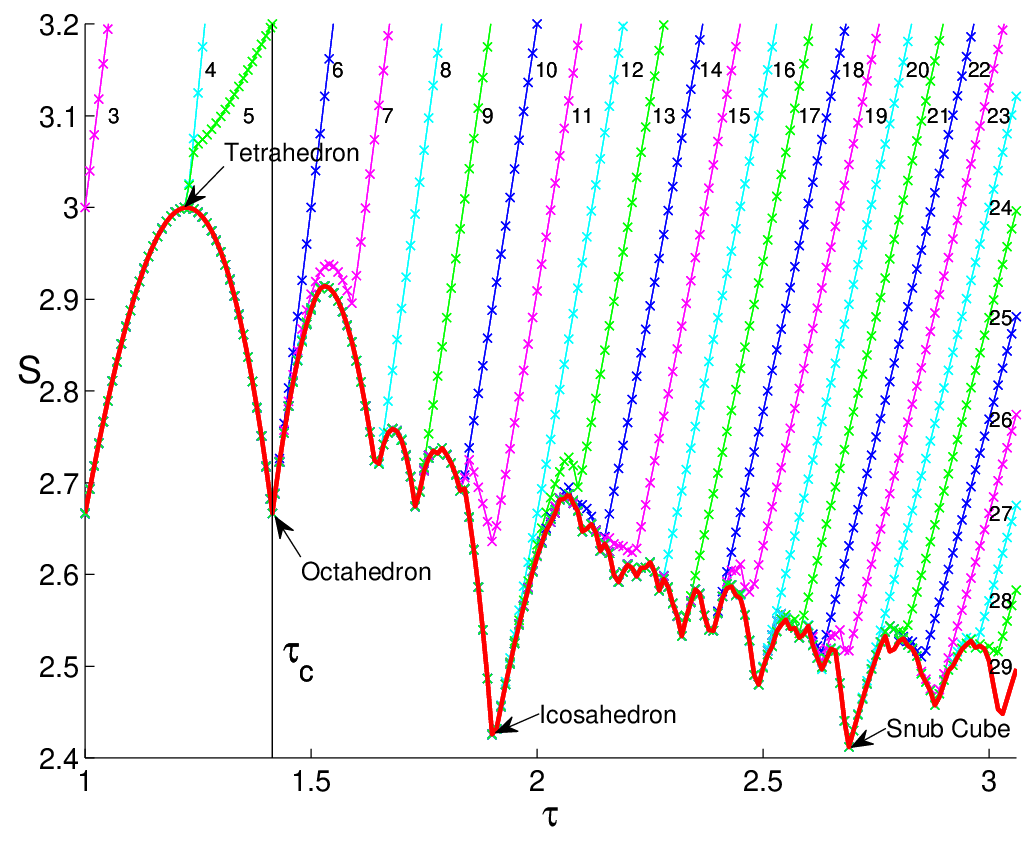}
 \caption{Numerical minima for the weighted counting measure on the sphere.}
 \label{figvgl_weight}
 \end{figure}

These plots suggest the following structure of the minimizers.
Let us denote the minimizing weighted counting measure for given~$m$ by~$\rho(m)$.
Then for any fixed~$\tau$, the series~$\Sact[\rho(m)]$ is monotone decreasing
(this is obvious because every~$\rho(m)$ can be realized by a weighted counting measure
with~$m_+>m$ summands by choosing $m_+-m$ weights equal to zero).
The important observation is that there is an integer~$m_0$ from which point on
the series stays constant, i.e.
\[ \Sact[\rho(m_-)] > \Sact[\rho(m_0)] = \Sact[\rho(m_+)] \qquad \forall\; m_- < m_0 < m_+\:. \]
This implies that the measure~$\rho_{m_0}$ is also a minimizer in the class
of all Borel measures. This leads us to the following
\begin{Conjecture} For any~$\tau \geq 1$, there is a minimizer~$\rho \in \M$ of the
variational problem on the sphere which is a weighted counting measure
supported at~$m_0$ points.
\end{Conjecture} \noindent
From Figure~\ref{figvgl_weight} we can read off the value of~$m_0$ as a function of~$\tau$.
More precisely, if~$m<m_0$, then~$\Sact[\rho(m)]$ is not minimal and coincides with the action for the
Tammes distribution~\eqref{Tammes}. The corresponding curves are labelled in Figure~\ref{figvgl_weight}
by~$m=3,4, \ldots, 29$. If~$m \geq m_0$, on the other hand, $\Sact[\rho(m)]$
is minimal and lies on the thick curve in Figure~\ref{figvgl_weight}.
Obviously, $m_0$ increases as~$\tau$ gets larger. This corresponds to the fact that
for increasing~$\tau$, the opening angle~$\vartheta_{\max}$ of the light cones gets smaller,
so that it becomes possible to distribute more points on the sphere which are all
spatially separated from all the other points.

The more detailed numerical study of the minimizers showed another interesting effect.
For values~$\tau < \tau_c := \sqrt{2}$, we found many different minimizers of different
form. They all have the property that they are {\em{completely timelike}} in the sense
that all points in the support of the minimizing measure have timelike or lightlike separation from all
the other points. We found minimizers supported on an arbitrarily large number of
points. If, on the other hand~$\tau > \tau_c$, all minimizers were supported on at most~$m_0(\tau)$
points, indicating that every minimizing measure~$\rho \in \M$ should be {\em{discrete
with finite support}}. The intermediate value~$\tau = \tau_c$ corresponds to
the opening angle~$\vartheta_{\max}=\frac{\pi}{2}$ of the light cones.
\begin{Conjecture} If~$\tau< \tau_c$, every minimizer is completely timelike.
If, conversely, $\tau > \tau_c$, every minimizing measure is discrete with finite support.
\end{Conjecture} \noindent
More graphically, one can say that for~$\tau > \tau_c$, our variational principle
{\em{spontaneously generates a discrete structure}} on the sphere.
The two regions~$\tau<\tau_c$ and~$\tau > \tau_c$ can also be understood
as two different phases of the system, so that at~$\tau = \tau_c$ we have a {\em{phase transition}}
from the completely timelike phase to the discrete phase.

The above numerical results will serve as the guideline for our analysis.
More precisely, the completely timelike phase will be analyzed in Section~\ref{secgentime}
using the notion of  ``generically timelike'',
whereas in Section~\ref{secsingular} we will develop under which assumptions and in which sense
the support of the minimizing measure is discrete or ``singular''.
The phase transition is made precise in Theorem~\ref{thmmain1} and~\ref{thmmain2} by stating that
minimizing measures are either generically timelike or singular.

\section{General Structural Results} \label{sec4}
We now return to the general variational principle~\eqref{var} with the Lagrangian
of the form~\eqref{Lform} and~\eqref{Ddef} on a general smooth compact manifold~$\F$.
Let us introduce some notation. For a given measure~$\rho \in \M$,
we define the functions
\begin{align}
\ell(x) &= \int_\F \L(x,y)\: d\rho(y) \;\; \in C^{0,1}(\F) \label{ldef} \\
\dd(x) &= \int_\F \D(x,y)\: d\rho(y) \;\; \in C^\infty(\F)\:. \label{ddef}
\end{align}
Moreover, we denote the Hilbert space~$L^2(\F, d\rho)$ by~$(\Hil_\rho, \lbra .,. \lket_\rho)$
and introduce the operators
\begin{align*}
\L_\rho \::\: \Hil_\rho \rightarrow \Hil_\rho \::\: \psi \mapsto (\L_\rho \psi)(x) &=
\int_\F \L(x,y)\: \psi(y)\: d\rho(y) \\
\D_\rho \::\: \Hil_\rho \rightarrow \Hil_\rho \:: \psi \mapsto (\D_\rho \psi)(x) &=
\int_\F \D(x,y)\: \psi(y)\: d\rho(y)
\end{align*}
(we use the consistent notation that a subscript~$\rho$ always denotes the corresponding operator on~$\Hil_\rho$).

\begin{Lemma} \label{lemmacompact}
The operators~$\L_\rho$ and~$\D_\rho$ are self-adjoint and Hilbert-Schmidt.
The eigenfunctions of~$\L_\rho$ (and $D_\rho$) corresponding to the non-zero eigenvalues
can be extended to Lipschitz-continuous (respectively smooth) functions on~$\F$.
\end{Lemma}
\Proof We only consider~$\D_\rho$, as the proof for~$\L_\rho$ is analogous.
The self-adjointness follows immediately from the fact that~$\D(x,y)$
is symmetric. Moreover, as the kernel is smooth and~$\F$ is compact, we know that
\[ \iint_{\F \times \F} |\D(x,y)|^2 d\rho(x)\, d\rho(y) < \infty\:. \]
This implies that~$\D_\rho$ is Hilbert-Schmidt (see~\cite[Theorem~2 in Section~16.1]{lax}).

Suppose that~$\D_\rho \psi = \lambda \psi$ with~$\lambda \neq 0$. Then the representation
\[ \psi(x) = \frac{1}{\lambda} \int_\F \D(x,y)\: \psi(y)\: d\rho(y) \]
shows that~$\psi \in C^\infty(\F)$ (recall that~$\D$ is smooth according to~\eqref{Ddef}).
\QED

The following notions characterize properties of~$\F$ and the function~$\D$
that will be needed later on.
\begin{Def} \label{defhomogen}
A measure~$\mu \in \M$ is a {\bf{homogenizer}} of~$\D$ if $\supp \mu = \F$ and both functions
\[ \ell_{[\mu]}(x) := \int_\F \L(x,y) \:d\mu(y) \qquad \text{and} \qquad
\dd_{[\mu]}(x) := \int_\F \D(x,y) \:d\mu(y)  \]
are constant on $\F$.
The function~$\D$ is called {\bf{homogenizable}} if a homogenizer exists.
\end{Def} \noindent
In Examples~(a) and~(b) in Section~\ref{secmath}, we can always choose
the standard normalized volume measure as the homogenizer.
More generally, in Example~(c) we choose for~$\mu$ the normalized Haar measure,
obtained by introducing a $\U(f)$-invariant metric on~$\F$ and taking the corresponding
volume form (see for example~\cite[Section~I.5]{broecker+tomdieck}).

The next proposition gives a sufficient condition for a homogenizer to be a minimizer.
\begin{Prp}\label{prp_hom_min}
If~$\Ll_\mu \geq 0$, the homogenizer~$\mu$ is a minimizer of the variational principle~\eqref{var}.
\end{Prp}
\begin{proof} We denote the constant function on~$\F$  by~$1_\F \equiv 1$.
If $\mu$ is a homogenizer, 
this function is an eigenfunction of $\L_\mu$, which can be completed to an orthonormal
eigenvector basis~$(\psi_i)_{i \in \N_0}$ of~$\Hil_\mu$ with~$\psi_0 = 1_\F$ and
corresponding eigenvalues $\lambda_i \geq 0$.

Using an approximation argument in the $C^0(\Ff)^*$-topology, it suffices to show that
\[ \Ss[\mu] \leq \Ss[\psi \mu] \]
for any~$\psi \in C^0(\F)$ with~$\psi \geq 0$ and~$\lbra \psi,1_\F\lket_\mu=1$.
We write~$\psi$ in the eigenvector basis~$\psi_i$,
\[ \psi = \sum_{i=0}^\infty c_i \, \psi_i \:. \]
The condition $\lbra\psi,1_\F \lket_\mu=1$ implies that~$c_0=1$. Thus 
$$ \Ss[\psi \mu] = \lbra\psi,\Ll_{\mu} \psi \lket_\mu=\lambda_0 + \sum_{i=1}^\infty
|c_i|^2 \lambda_i\geq\lambda_0 = \Ss[\mu]\:. $$
\vspace*{-2.5em}

\end{proof}

\subsection{The Euler-Lagrange Equations}
Let us assume that~$\rho$ is a minimizer of the variational principle~\eqref{var},
\[ \Ss[\rho] = \inf_{\tilde{\rho} \in \M} \Ss[\tilde{\rho}] =: \Smin \:. \]
We now
derive consequences of minimality. In the first lemma, we consider first variations of~$\rho$
to obtain the Euler-Lagrange equations corresponding to our variational principle.
The second lemma, on the other hand, accounts for a nonlinear effect.

\begin{Lemma} (The Euler-Lagrange equations) \label{lemmaEL}
\beq \label{EL1}
\ell|_{\supp \rho} \,\equiv\, \inf_\F \ell = \Smin \:.
\eeq
\end{Lemma}
\Proof Comparing~\eqref{Sdef} with~\eqref{ldef}, one sees that
\beq \label{Srep}
\Smin = \int_\F \ell \: d\rho \:.
\eeq

Since~$\ell$ is continuous and~$\F$ is compact, there clearly is~$y \in \F$ with
\beq \label{ly}
\ell(y) = \inf_\F \ell \:.
\eeq
We consider for~$t \in [0,1]$ the family of measures
\[ \tilde{\rho}_t = (1-t)\, \rho + t \, \delta_y \; \in \M \:, \]
where~$\delta_y$ denotes the Dirac measure supported at~$y$. Substituting this formula in~\eqref{Sdef}
and differentiating, we obtain for the first variation the formula
\[ \delta \Sact := \lim_{t \searrow 0} \frac{\Sact[\tilde{\rho}_t] - \Sact[\tilde{\rho}_0]}{t}
= -2 \Smin + 2 \ell(y)\:. \]
Since~$\rho$ is a minimizer, $\delta \Sact$ is positive. Combining this result with~\eqref{Srep}
and~\eqref{ly}, we obtain the relations
\[ \inf_\F \ell = \ell(y) \:\geq\: \Smin = \int_\F \ell\: d\rho \:. \]
It follows that~$\ell$ is constant on the support of~$\rho$, giving the result.
\QED
The statement of Lemma~\ref{lemmaEL} implies that for any~$x\in \supp\rho$,
\beq \label{1eigen}
(\L_\rho 1_\F)(x) = \int_\F \L(x,y)\: d\rho(y) = \ell(x) = \Smin\: 1_\F(x)\:,
\eeq
showing that the constant function~$1_\F$ is an eigenvector of the operator~$\L_\rho$ corresponding
to the eigenvalue~$\Smin \geq 0$. 
However, the statement~\eqref{EL1} is stronger because it gives information on~$\ell$ even
away from the support of~$\rho$.

\begin{Lemma} \label{lemmaP}
The operator~$\L_\rho$ is positive (but not necessarily strictly positive).
\end{Lemma}
\Proof Assume that the lemma is wrong. Then, as~$\L_\rho$ is a compact and self-adjoint
operator (see Lemma~\ref{lemmacompact}), there exists an eigenvector~$\psi$ corresponding to
a negative eigenvalue, $\L_\rho \psi = \lambda \psi$ with~$\lambda<0$.
We consider the family of measures
\[ \tilde{\rho}_t = (1_\F + t \psi) \, \rho\:. \]
In view of Lemma~\ref{lemmacompact}, $\psi$ is continuous and therefore bounded.
Thus for sufficiently small~$|t|$, the measure~$\tilde{\rho}_t$ is positive.
In view of~\eqref{1eigen}, the functions~$1_\F$ and~$\psi$ are eigenfunctions
corresponding to different eigenvalues. Hence these eigenfunctions are orthogonal.
Thus
\[ \tilde{\rho}_t(\F) = \int_F 1_\F \,(1_\F + t \psi)\: d\rho
= 1 +t\, \lbra 1_\F, \psi \lket_\rho = 1 \:, \]
showing that~$\tilde{\rho}_t$ is again normalized. Finally, again using the
orthogonality,
\[ \Sact[\tilde{\rho}_t] = \lbra (1_\F + t \psi), \L_\rho (1_\F + t \psi) \lket_\rho
= \Smin + \lambda \: t^2\: \lbra \psi, \psi \lket_\rho \:. \]
Thus~$\tilde{\rho}_t$ is an admissible variation which decreases the action, a contradiction.
\QED

This lemma has useful consequences. We first derive a positivity property
of the Lagrangian when evaluated on a finite number of points in the support of~$\rho$.
\begin{Corollary}\label{matrix}
For a finite family~$x_0, \ldots, x_N \in \supp \rho$ (with~$N \in \N$), the \textbf{Gram matrix}~$L$ defined by
\[ L=\Big(\Ll(x_i,x_j)\Big)_{i,j=0,\ldots,N} \]
is symmetric and positive semi-definite.
\end{Corollary}
\begin{proof} Given~$\varepsilon>0$ and a vector~$u = (u_0, \ldots, u_N) \in \C^{N+1}$, we set
$$\psi_\epsilon (x)=\sum_{i=0}^N \frac{u_i}{\rho(B_\varepsilon(x_i))} \: \chi_{B_\varepsilon(x_i)}(x)\in \Hil_\rho\:, $$
where~$B_\varepsilon$ is a ball of radius~$\varepsilon$ (in any given coordinate system).
Lemma~\ref{lemmaP} implies that~$\lbra \psi_\varepsilon, \Ll_\rho \psi_\varepsilon\lket \geq 0$.
Taking the limit~$\varepsilon \searrow 0$, it follows that
$$ \lbra u ,L u \lket_{\C^{N+1}} = \lim_{\varepsilon \searrow 0} \lbra \psi_\varepsilon, \Ll_\rho 
\psi_\varepsilon\lket_\rho \geq 0 \:.$$
\vspace*{-2.1em}

\end{proof} \noindent
We next derive a simple criterion which guarantees that a minimizing measure~$\rho$
cannot be supported on the whole manifold~$\F$.
\begin{Corollary} Assume that there is a bounded regular Borel measure~$\nu$ 
(not necessarily positive) such that
\[ \iint_{\F \times \F} \L(x,y)\, d\nu(x)\, d\nu(y) < 0 \:. \]
Then the support of a minimizing measure cannot contain the support of~$\nu$,
\[ \supp \rho \not \supset \supp \nu \:. \]
\end{Corollary}
\Proof If we assume, conversely, that~$\supp \rho \supset \supp \nu$, the measure~$\nu$
can be approximated by measures of the form~$\psi_k \,\rho$ with~$\psi_k \in \Hil_\rho$
(in the sense that~$\psi_k\, \rho \rightarrow \nu$ in the weak-$C^0(\F)^*$-topology).
As a consequence, the expectation value~$\lbra \psi_k, \L_\rho \psi_k \lket$ is negative for large~$k$,
in contradiction to Lemma~\ref{lemmaP}.
\QED
This corollary explains why minimizers have the tendency of being supported on
proper subsets of~$\F$. But the argument is too weak for concluding discreteness.
In order to get more detailed information on the form of the minimizing measures, we
need more advanced notions and methods, which we now introduce.

\subsection{Generically Timelike Minimizers} \label{secgentime}

\begin{Def} \label{defgtl}
A minimizing measure~$\rho \in \M$ is called {\bf{generically timelike}}
if the following conditions hold:
\begin{itemize}
\item[(i)] $\D(x,y) \geq 0$ for all~$x,y \in \text{supp}\, \rho$.
\item[(ii)] The function~$\dd$ defined by~\eqref{ddef} is constant on~$\F$.
\end{itemize}
\end{Def} \noindent
This constant can easily be computed:
\begin{Lemma} \label{lemma48} Suppose that~$\rho$ is a generically timelike minimizer. Then
\[ \dd(x) = \Smin \qquad \text{for all~$x \in \F$}\:. \]
\end{Lemma}
\Proof In view of property~(i), $\L$ and~$\D$ coincide on the support of~$\rho$. Thus
\[ \Smin = \iint_{\F \times \F} \L(x,y)\: d\rho(x)\: d\rho(y) \\
= \iint_{\F \times \F} \D(x,y)\: d\rho(x)\: d\rho(y) \:. \]
Carrying out one integral using~\eqref{ddef}, we obtain
\[ \Smin = \int_\F \dd(x)\: d\rho(x) \:. \]
Using property~(ii) gives the result.
\QED

In the remainder of this subsection, we assume that~$\D$ is homogenizable
(see Definition~\ref{defhomogen}) and denote the homogenizer by~$\mu \in \M$. 
\begin{Lemma} If $\D_{\mu}$ has only a finite number of negative eigenvalues, then the kernel~$\D(x,y)$
has the representation
\beq \label{Dspec}
\D(x,y) = \nu_0 + \sum_{n=1}^N \nu_n\: \phi_n(x)\: \overline{\phi_n(y)}
\eeq
with~$N \in \N \cup \{\infty \}$, $\nu_n \in \R$, $\nu_n\neq 0$, and~$\phi_n \in C^\infty(\F)$, where
in the case~$N=\infty$ the series converges uniformly.
\end{Lemma}
\Proof By definition of the homogenizer, the function~$1_\F \equiv 1$
is an eigenfunction of the operator~$\D_\mu$.
Denoting the corresponding eigenvalue by~$\nu_0$, we obtain the spectral representation~\eqref{Dspec}.

If $\D_\mu$ is positive, the uniform convergence is an immediate generalization of
Mercer's theorem (see \cite[Theorem~11 in Chapter~30]{lax}, where we replace the interval~$[0,1]$ by
the compact space~$\F$, and the Lebesgue measure by the measure~$\mu$).
In the case when~$\D_\mu$ has a finite number of negative eigenvalues, we apply Mercer's theorem
similarly to the operator with kernel $\D(x,y)-\sum_{i=1}^K \lambda_i\,
\psi_i(x)\overline{\psi_i(y)}$, where~$\lambda_1, \ldots \lambda_K$ are the negative eigenvalues
with corresponding eigenfunctions~$\psi_i$. By construction, this operator is positive,
and in view of Lemma~\ref{lemmacompact} its kernel is continuous.
\QED

\begin{Lemma} \label{lemmaspec}
Suppose that~$\rho$ is a generically timelike minimizer
and that the operator~$\D_{\mu}$ has only a finite number of negative eigenvalues. Then
\[ \Sact[\rho] =  \nu_0 \qquad \text{and} \qquad
\int_\F \phi_n(y) \:d\rho(y) = 0 \quad \text{for all~$n \in \{1, \ldots, N\}$} \:. \]
\end{Lemma}
\Proof Using the decomposition of the kernel~\eqref{Dspec} and the uniform convergence, we obtain
\beq \label{drep}
\dd(x)=\nu_0+\sum_{n=1}^N \nu_n\: \phi_n(x)\:\int_\F \overline{\phi_n(y)} \,d \rho(y)\:.
\eeq
Applying Lemma~\ref{lemma48} gives the claim.	
\QED

\begin{Prp}\label{est_nu}
Suppose that~$\D_\mu$ is a positive operator on the Hilbert space~$\Hil_\mu$.
Then
\[ \Smin \geq \nu_0 \:. \]
In the case of equality, every minimizer is generically timelike.
\end{Prp}
\Proof If~$\D_\mu$ is positive, all the parameters~$\nu_n$ in~\eqref{Dspec}
are positive. It follows that for every measure~$\tilde{\rho} \in \M$,
\beq \label{Sest}
\Sact[\tilde{\rho}] = \iint_{\F \times \F} \Ll(x,y)\: d\tilde{\rho}(x)\: d\tilde{\rho}(y)
\geq \iint_{\F \times \F} \D(x,y)\: d\tilde{\rho}(x)\: d\tilde{\rho}(y) \\
\geq \nu_0\: \tilde{\rho}(\F)^2 = \nu_0 \:.
\eeq
Let us assume that equality holds. It then follows from~\eqref{Sest} that~$\Ll$ and~$\D$
coincide on the support of~$\tilde{\rho}$ and thus~$\D(x,y) \geq 0$ for all~$x, y \in \supp \tilde{\rho}$.
Moreover, we find from~\eqref{Dspec} or~\eqref{drep} that
\[ \nu_0 = \nu_0 + \sum_{n=1}^N \nu_n
\left| \int_{\F} \overline{\phi_n(y)}\: d\tilde{\rho} \right|^2 \]
and thus
\[ \int_{\F} \overline{\phi_n(y)} \: d\tilde{\rho} = 0 \qquad \text{for all~$n \geq 1$}\:. \]
It follows that~$\dd_{\tilde{\rho}}$ is a constant. We conclude that~$\tilde{\rho}$ is generically timelike.
\QED
This proposition can be used to construct generically timelike minimizers.
\begin{Corollary} \label{cor412}
Suppose that~$\D_\mu$ is a positive operator on~$\Hil_\mu$.
Assume that the function~$f\in \Hil_\mu$ has the following properties:
\begin{itemize}
\item[(a)] $\D(x,y)=\L(x,y)$ for all $x,y\in \supp f$.\\[-0.5em]
\item[(b)] $\displaystyle{ \int_\F \,f(x) \,d \mu(x)=1 \quad \text{and} \quad
\int_\F f(x) \,\phi_n(x) \,d \mu(x)=0 \quad \text{for all $n\in\{1,\ldots,N\}$}}.$
\end{itemize}
Then the measure~$d\rho = f\, d\mu$ is a generically timelike minimizer.
\end{Corollary}
\Proof The assumption~(a) implies that
\[ \Sact[\rho] = \iint_{\F \times \F} \D(x,y)\: d\rho(x)\: d\rho(y) \:. \]
Using the decomposition~\eqref{Dspec} and the relations~(b),
we find that~$\Sact[\rho]=\nu_0$. We now apply Proposition~\ref{est_nu}.
\QED

We conclude this section by stating obstructions for the existence of generically timelike minimizers.
\begin{Prp}\label{Prpgt} Assume that one of the following conditions holds:
\begin{itemize}
\item[(I)] The operator~$\D_{\mu}$ has only a finite number of negative eigenvalues,
and the eigenvalue~$\nu_0$ in the decomposition~\eqref{Dspec} is non-positive.
\item[(II)] For every~$x \in \F$ there is a point~$y \in \F$ with~$\J(x) \cap \J(y) = \varnothing$
(``condition of disjoint lightcones'').
\item[(III)] For every~$x \in \F$ there is a point~$-x \notin \overline{\I(x)}$ with~$\J(x)=\overline{\I(x)}\cup \{-x\}$ and $\overline{\I(x)}\cap \overline{\I(-x)}=\varnothing$
(``condition of antipodal points'').
\end{itemize}
Then there are no generically timelike minimizers.
\end{Prp}
\Proof We first show that~$\Smin > 0$. Namely, choosing~$x$ in the support of a minimizing
measure~$\rho$, we know from~\eqref{posdiagonal} and the continuity of~$\D$ that
there is a neighborhood~$U$ of~$x$ and~$\delta>0$ such that~$\D(x,y) > \delta$ for all~$y \in U$. It
follows that
\[ \Smin \geq \int_{U \times U} \L(x,y)\: d\rho(x) \, d\rho(y) \geq \delta \, \rho(U)^2 > 0 \:. \]

Case~(I) is obvious in view of Lemma~\ref{lemmaspec} and the fact that~$\Smin > 0$.
To prove the remaining cases~(II) and~(III), we assume conversely that there exists a generically timelike minimizer~$\rho \in \M$. Choosing a point~$x \in \supp \rho$, we know from property~(i)
in Definition~\ref{defgtl} that~$\supp \rho \subset \J(x)$. In case~(II), we choose~$y \in \F$
with~$\J(x) \cap \J(y) = \varnothing$ to obtain
\[ \dd(y)=\int_{\J(x)} \D(y,z)\,d \rho(z)\leq 0<\Smin \:, \]
in contradiction to Lemma~\ref{lemma48}.

In case~(III), we know that $\supp \rho \subset \J(x)
= \overline{\I(x)} \cup \{-x\}$. If~$-x\notin \supp \rho$, the estimate
\[ \dd(-x)=\int_{\J(x)} \D(-x,z)\,d \rho(z) = \int_{\overline{\I(x)}} \D(-x,z)\,d \rho(z)
\overset{(*)}{\leq} 0<\Smin \]
again gives a contradiction, where in~(*) we used that~$\overline{\I(x)}\cap \overline{\I(-x)}=\varnothing$.
If conversely~$-x\in \supp \rho$,
then $\supp \rho \subset \J(x) \cap \J(-x) = \{x\} \cup \{-x\}$
(where we again used that~$\overline{\I(x)}\cap \overline{\I(-x)}=\varnothing$).
Hence the integral in~\eqref{ddef} reduces to a sum over two points,
\beq \label{drel}
\dd(y) = \rho( \{x \})\: \D(y,x) + \rho( \{-x \})\: \D(y,-x)\:.
\eeq
In view of our assumption~\eqref{posdiagonal}, we know that~$x \in \I(x)$ and~$-x \in \I(-x)$.
On the other hand, the relation~$\overline{\I(x)}\cap \overline{\I(-x)}=\varnothing$ shows
that~$-x \notin \I(x)$. Hence there is a point~$y \in \partial \I(x)$.
It follows that~$\D(y,x)=0$ (because~$y \in \partial \I(x)$) and also~$\D(y,-x) \leq 0$
(because~$y \in \overline{\I(x)}$ and thus~$y \notin \overline{\I(-x)}$).
Using these inequalities in~\eqref{drel}, we again find that~$\dd(y) \leq 0$, a contradiction.
\QED
It is interesting to ask what the support of a generically timelike minimizer~$\rho$ may look like.
The next proposition (which will not be used later on) quantifies that~$\supp \rho$
must be ``sufficiently spread out.''

\begin{Prp} \label{lemmaiff}
Assume that~$\rho$ is a generically timelike minimizer
and that the operator~$\D_{\mu}$ has only a finite number of negative eigenvalues.
Then every real function~$\psi \in \D_\mu(\Hil_\mu)$ with
\beq \label{intc}
\int_\F \psi(x)\: d\mu(x) = 0
\eeq
changes its sign on the support of~$\rho$ (here~$\mu$ is again the homogenizer
of Definition~\ref{defhomogen}).
\end{Prp}
\Proof We return to the spectral decomposition~\eqref{Dspec} of the operator~$\D_\mu$.
Since the eigenfunctions~$\phi_n$ are orthogonal in~$\Hil_\mu$, we know that
\[ \int_\F \phi_n \: d\mu = 0 \qquad \text{for all $n \geq 1$}\:. \]
Representing~$\psi$ in an eigenvector basis of~$\D_\mu$ and using~\eqref{intc}, we find
$$\psi=\sum_{n=1}^N\kappa_n \, \phi_n$$ with complex coefficients~$\kappa_n$.
Integrating with respect to~$\rho$, we can apply Lemma~\ref{lemmaspec} to obtain
\[ \int_{\F} \psi(x)\,d \rho(x) = \sum_{n=1}^N \kappa_n \int_\F \phi_n(x)\,d \rho(x)=0 \:. \]
Hence~$\psi$ changes its sign on the support of~$\rho$.
\QED

\subsection{Minimizers with Singular Support} \label{secsingular}
We now state results on the support of a minimizing measure.

\begin{Thm} \label{thmmain1} 
Let~$\F$ be a smooth compact manifold. Assume that~$\D(x,y)$
is symmetric~\eqref{Ddef} and equal to one on the diagonal, $\D(x,x)\equiv 1$.
Furthermore, we assume that for every~$x \in \F$ and~$y \in \K(x)$, there is a smooth curve~$c$ joining
the points~$x$ and~$y$, along which~$\D(.,y)$ has a non-zero derivative at~$x$, i.e.
\beq \label{gradient}
\frac{d}{dt} \D \big( c(t), y \big) \Big|_{t=0} \neq 0 \:,
\eeq
where we parametrized the curve such that~$c(0)=x$.
Then the following statements are true:
\begin{itemize}
\item[(A)] If~$\F$, $\D$ are real analytic, then a minimizing measure~$\rho$
is either generically timelike or~$\overset{\circ}{\supp \rho} = \varnothing$.
\item[(B)] If~$\D$ is smooth and if there is a differential operator~$\Delta$ (of any finite order)
on~$C^\infty(\F)$ which vanishes on the constant functions such that
\beq \label{lapcond}
\Delta_x\D(x,y) <0 \qquad \text{for all~$y\in \I(x)$} \:,
\eeq
then~$\overset{\circ}{\supp \rho} = \varnothing$.
\end{itemize}
\end{Thm} \noindent
A typical example for~$\Delta$ is the Laplacian corresponding to
a Riemannian metric on~$\F$.
Note that the condition~\eqref{gradient} implies that for every~$y \in \F$,
the set~$\{x \:|\: y \in \K(x)\}$ is a smooth
hypersurface, which the curve~$c$ intersects transversely
(in the applications of Section~\ref{sec5} and~\ref{sec6}, this set will
coincide with~$\K(y)$, but this does not need to be true in general).

The condition~\eqref{gradient} can be removed if, instead, we make the following symmetry
assumption.
\begin{Def}
The function $\D$ is called \textbf{locally translation symmetric} at~$x$ with respect to a curve~$c(t)$
with~$c(0)=x$ if there is~$\varepsilon>0$ and a
function~$f \in C^\infty((-2 \varepsilon, 2 \varepsilon))$ such that the curve~$c$ is defined on the
interval~$(-\varepsilon, \varepsilon)$ and
\[ \D(c(t), c(t')) = f(t-t') \qquad \text{for all~$t,t' \in (-\varepsilon, \varepsilon)$} \:. \]
\end{Def} 
\begin{Thm} \label{thmmain2} 
Let~$\F$ be a smooth compact manifold. Assume that~$\D(x,y)$
is symmetric~\eqref{Ddef} and strictly positive on the diagonal~\eqref{posdiagonal}.
Furthermore, we assume that for every~$x \in \F$ and~$y \in \K(x)$, there is a smooth curve~$c$ joining
the points~$x$ and~$y$ such that~$\D$ is locally translation symmetric at~$x$ with respect to~$c$,
and such that the function~$\D(c(t), y)$ changes sign at~$t=0$
(where we again parametrize the curve such that~$c(0)=x$).
Then statement~(A) of Theorem~\ref{thmmain1} holds, provided that the curve~$c$ is
analytic in a neighborhood of~$t=0$. Assume, furthermore, that
there is~$p \in \N$ with
\beq \label{gradient_eq}
\frac{d^p}{dt^p} \D \big( c(t), y \big) \Big|_{t=0} \neq 0 \:.
\eeq
Then statement~(B) of Theorem~\ref{thmmain1} again holds.
\end{Thm} \noindent

In the smooth setting, the above theorems involve quite strong additional
assumptions (see~\eqref{gradient}, \eqref{lapcond} and~\eqref{gradient_eq}).
The following counter example shows that some conditions of this type
are necessary for the statements of these theorems to be true\footnote{We would like to thank Robert
Seiringer for pointing out a similar example to us.}.
\begin{Example} {\em{
Let~$f, g \in C^\infty_0([-\pi, \pi])$ be non-negative even functions with
\[ \supp f \subset \left[-\tfrac{\pi}{8}, \tfrac{\pi}{8} \right] \:,\qquad
\supp g \subset \left(- \pi, -\tfrac{\pi}{2} \right] \cup \left[ \tfrac{\pi}{2}, \pi \right)\:. \]
We introduce the function~$\D \in C^\infty(S^2 \times S^2)$ by
\beq \label{Ddefex}
\D(x,y)= -g \big( \dist(x,y) \big) + \int_{S^2} f \big( \dist(x,z) \big)\: f \big( \dist(z,y) \big) \:d \mu(z) \:,
\eeq
where~$d\mu$ is the standard volume measure, and~$\dist$ denotes the geodesic distance
(taking values in~$[0, \pi]$).
Note that the two summands in~\eqref{Ddefex} have disjoint supports, and thus the corresponding
Lagrangian~\eqref{Lform} is simply
\[ \L(x,y)= \int_{S^2} f \big( \dist(x,z) \big)\: f \big( \dist(z,y) \big) \:d \mu(z) \:. \]
We again consider~$\D(x,y)$ and~$\L(x,y)$ as the integral kernels of corresponding operators~$\D_\mu$
and~$\L_\mu$ on the Hilbert space~$\Hil_\mu = L^2(S^2, d\mu)$.

First, it is obvious that~$\D(x,y)$ is symmetric and constant on the diagonal. Next, it is clear
by symmetry that the measure~$\mu$ is a homogenizer (see Definition~\ref{defhomogen}).
Moreover, writing~$\L_\mu$ as~$\L_\mu = f_\mu^2$, where~$f_\mu$ is the operator with
integral kernel~$f$, one sees that the operator~$\L_\mu$ is positive.
Thus by Proposition~\ref{prp_hom_min}, the measure~$\mu$ is minimizing.
If the function~$g$ is non-trivial, there are points~$x, y$ which are spacelike separated, so that
this minimizer is not generically timelike. Also, its support obviously has an empty interior.
We have thus found a minimizing measure which violates statement~(A) of Theorem~\ref{thmmain1}.
}} \QEDrem
\end{Example}

The remainder of this section is devoted to the proof of the above theorems.
We begin with a simple but very useful consideration. Suppose that for given~$x \in \F$,
the boundary of the light cone~$\K(x)$ does not intersect the support of~$\rho$.
As the support of~$\rho$ is compact, there is neighborhood~$U$ of~$x$ such that
\[ \K(z) \cap \supp \rho = \varnothing \qquad \text{for all~$z \in U$}\:. \]
Thus introducing the measure~$\hat{\rho} = \chi_{\I(x)}\, \rho$, the function~$\ell$
can for all~$z \in U$ be represented by
\beq \label{ellrel}
\ell(z) = \int_{\F} \Ll(z,\xi)\, d\hat{\rho}(\xi) 
= \int_{\F} \D(z,\xi)\, d\hat{\rho}(\xi) \:.
\eeq
In the following two lemmas, we make use of this identity in the smooth and analytic settings.
\begin{Lemma}
If~\eqref{lapcond} holds, then for every~$x \in \supp \rho$ the set~$\K(x) \cap \supp \rho$ is nonempty.
\end{Lemma}
\Proof Applying the differential operator~$\Delta_x$ to~\eqref{ellrel} gives
\[ \Delta_x \ell(x) = \int_{\F} \Delta_x \D(x,z)\, d\hat{\rho}(z) <0\:, \]
where in the last step we used~\eqref{lapcond} and the fact that~$x \in \supp \rho$.
This is a contradiction to Lemma~\ref{lemmaEL}.
\QED

\begin{Lemma}
Suppose that~$\F$ and~$\D$ are real analytic. Assume that there exists a
point~$x \in \overset{\circ}{\supp \rho}$
such that~$\K(x) \cap \supp \rho = \varnothing$. Then~$\rho$ is generically timelike
and~$\supp\rho \subset \I(x)$.
\end{Lemma}
\begin{proof} We introduce on~$\F$ the function
\[  \hat{\dd}(y)=\int_\Ff \D(y,z) \:d\hat{\rho}(z)\:. \]
Then~$\hat{\dd}$ is real analytic and, according to~\eqref{ellrel}, it coincides on~$U$ with
the function~$\ell$. Since~$x \in \overset{\circ}{\supp \rho}$, the Euler-Lagrange equations in
Lemma~\eqref{lemmaEL} yield that~$\ell \equiv \Smin$ in a neighborhood of~$x$.
Hence~$\hat{\dd}\equiv \Smin$ in a neighborhood of~$x$, and the real analyticity implies that
\[ \hat{\dd} \equiv \Smin \qquad \text{on~$\F$}\:. \]
It follows that
\beq \label{iestim}
\begin{split}
\Smin &=\int_\Ff \hat{\dd}(x) \:d\rho(x)=\iint_{\F \times \F}  \D(x,y) \:d \hat{\rho}(x) \:d\rho(y)\\
& \leq \iint_{\F \times \F}  \Ll(x,y) \:d \hat{\rho}(x) \:d\rho(y)
=\int_\F \ell(x)\:d\hat{\rho}(x) = \Smin\; \hat{\rho}(\Ff) \:,
\end{split}
\eeq
and thus~$\hat{\rho}(\Ff)=1$. Since~$\hat{\rho} \leq \rho$ and~$\rho$ is normalized,
we conclude that~$\rho=\hat{\rho}$. Thus~$\dd \equiv \hat{\dd} \equiv \Smin$.
Moreover, the inequality in~\eqref{iestim} becomes an equality,
showing that~$\Ll \equiv \D$ on the support of~$\rho$.
Thus~$\rho$ is indeed generically timelike.
\end{proof}

%\begin{Corollary}
%Assume that the conditions of Lemma \ref{real_anal} hold. If for $x\in \F$ there exists $y\in \F$ such that $\I(x)\cap \I(y)=\varnothing$, then $ \overset{\circ}{\supp \rho}=\varnothing$.
%\end{Corollary}
%\Proof
%In this case $\rho$ is generically timelike with $\supp\rho \subset \I(x)$. But for $y\in \Ff$ with $\I(x)\cap \I(y)=\varnothing$, one obtains $$\ell(y)=\int_{\I(x)} \Ll(y,z)d \rho(z)=0<\Smin$$ in contradiction to Lemma~\ref{lemmaEL}.
%\QED

To complete the proof of Theorems~\ref{thmmain1} and~\ref{thmmain2}, it remains to show the
following statement:
\beq \label{remains}
\K(x) \cap \supp \rho = \varnothing \qquad \text{for all~$x \in \overset{\circ}{\supp \rho}$}\:.
\eeq
We proceed indirectly and assume that there is a point~$y \in \K(x) \cap \supp \rho$.
Our strategy is to choose points~$x_0, \ldots, x_k$ in a neighborhood of~$x$
such that~$\L$ restricted to the set~$\{x_0, \ldots, x_k, y\}$ is not positive semi-definite,
in contradiction to Corollary~\ref{matrix}. The points~$x_0, \ldots, x_k$ will all lie on a fixed
smooth curve~$c$ which joins~$x$ and~$y$ chosen as in the statement of the theorems.
We parametrize~$c$ such that~$c(0)=x$ and~$c(1)=y$, and by extending the curve we 
can arrange that the curve is defined on the interval~$(-k \varepsilon, 1]$ for suitable~$\varepsilon>0$.
By the assumptions in Theorems~\ref{thmmain1}
and~\ref{thmmain2}, we know that~$\D(c(t), y)$ changes sign at~$t=0$. 
Depending on the sign of~$\D(c(\varepsilon), y)$,
we introduce the ``chain'' of points
\beq \label{chain}
\hspace*{-.3cm} \begin{cases}
x_0 = c(\varepsilon), \;\; x_1 = c(0), \;\ldots\;, x_k=c \big(-(k-1) \varepsilon \big)
& \text{if~$\D(c(\varepsilon), y)>0$} \\[0.3em]
x_0 = c(-\varepsilon), \;\; x_1 = c(0), \;\ldots\; , x_k=c \big( (k-1) \varepsilon \big)
& \text{if~$\D(c(\varepsilon), y)<0$}
\end{cases}
\eeq
(thus~$y$ has timelike separation from~$x_0$, lightlike separation from~$x_1=x$,
and spacelike separation from~$x_2, \ldots, x_k$).
Then by construction, $x_0 \in \I(y)$, whereas all the other points of the chain
are spacelike or lightlike separated from~$y$.

For the proof of Theorem~\ref{thmmain1}, it suffices to
consider a chain of three points.

\begin{Lemma} \label{lemma420}
Assume that~$\D(x,y)$
is symmetric~\eqref{Ddef} and equal to one on the diagonal, $\D(x,x)\equiv 1$.
Then for~$x_0,x_1,x_2$ as given by~\eqref{chain} in the case~$k=2$, there is a real constant~$a_1$
such that, for all sufficiently small~$\varepsilon$,
\begin{equation}\label{D_exp_3}
\D(x_i,x_j)=1 + a_1\: |i-j|^2  \:\varepsilon^2+\Oo(\varepsilon^3) \qquad 
\text{for all~$i,j \in \{0, 1, 2\}$}\:.
\end{equation}
\end{Lemma}
\begin{proof} We set~$f(t, t') = \D(c(t), c(t'))$ for~$t, t' \in (-2 \varepsilon, 2 \varepsilon)$.
Using that $\D$ is symmetric and that $\D(x,x)\equiv 1$, we know that
\[ 0 = \frac{d}{dt} f(t, t) \big|_{t=t_0} = 2 \:\frac{d}{dt} f(t_0, t) \big|_{t=t_0}\:. \]
Thus the linear term in a Taylor expansion vanishes,
\[ f(t_0, t) = 1 + \frac{1}{2}\: g(t_0)\:
(t-t_0)^2 + \Oo \big(|t-t_0|^3 \big) \:, \]
where we set
\[ g(t_0) = \frac{d^2}{dt^2} f(t_0, t)  \Big|_{t=t_0} \:. \]
As the function~$g$ is smooth, we can again expand it in a Taylor series,
\[ g(t_0) = g(0) + \Oo(t_0) \:. \]
We thus obtain
\[ f(t_0, t) = 1 + \frac{1}{2}\: g(0)\:
(t-t_0)^2 + \Oo \big( |t_0|\: |t-t_0|^2 \big) + \Oo \big( |t-t_0|^3 \big) \:. \]
Setting~$a_1 = g(0)/2$ and using that~$|t|, |t_0| \leq 2 \varepsilon$, the result follows.
\QED

\begin{Lemma}
Under the assumptions of Theorem~\ref{thmmain1}, the statement~\eqref{remains} holds.
\end{Lemma}
\Proof 
Assume, conversely, that for $x \in \overset{\circ}{\supp \rho}$  there is a point $y\in \supp\rho\cap\K(x)$.
We choose the chain~$x_0, x_1=x, x_2$ as in Lemma~\ref{lemma420}.
We use the notation of Corollary~\ref{matrix} in case~$N=3$, setting~$x_3=y$.
Choosing the vector~$u \in \C^4$ as~$u = (1,-2,1,0)$, we can apply Lemma~\ref{lemma420}
to obtain
\begin{align*}
\langle u, L u \rangle_{\C^4} = 6-4\, \D(x_0,x_1)+2\, \D(x_0,x_2)-4\, \D(x_1,x_2)=\Oo(\varepsilon^3)\:.
\end{align*}
Furthermore, using~\eqref{gradient}, we know that
$$\D(x_0,y)=b\, \varepsilon +\Oo(\varepsilon^2)$$ with $b\neq 0$. Thus, choosing~$u=(\alpha, -2 \alpha,
\alpha, \beta)$ with~$\alpha, \beta \in \R$, it is
\begin{align*}
\langle u, L u \rangle_{\C^4} =
\bigg\langle \bpm \alpha \\ \beta \epm,
\bpm \Oo(\varepsilon^3) & b \varepsilon +\Oo(\varepsilon^2) \\ b \varepsilon +\Oo(\varepsilon^2) & 1 \epm \bpm \alpha \\ \beta \epm \bigg\rangle_{\C^2}\:.
\end{align*}
For sufficiently small~$\varepsilon$, the matrix in this equation has a negative determinant,
in contradiction to Corollary~\ref{matrix}.
\QED
This completes the proof of Theorem~\ref{thmmain1}.

In order to finish the proof of Theorem~\ref{thmmain2}, we first remark that,
combining the symmetry of~$\D$ with the assumption that~$\D$ is locally translation symmetric
at~$x$ with respect to~$c$, we know that~$\D(c(t), c(t')) = f(|t-t'|)$. After rescaling, we can assume that $f(0)=1$. A Taylor expansion of~$f$ then yields
the following simplification and generalization of Lemma~\ref{lemma420},
\begin{equation}\label{expand_D}
\D(c(t),c(t'))=1+\sum_{i=1}^K a_i \,(t-t')^{2i} +\Oo \Big( (t-t')^{2(K+1)} \Big),
\end{equation}
where the real coefficients~$a_i$ only depend on~$\D$ and the curve~$c$.
\begin{Lemma}
Under the assumptions of Theorem~\ref{thmmain2}, the statement~\eqref{remains} holds.
\end{Lemma}
\Proof Let us first verify that in the real analytic case, there is a~$p$ such that~\eqref{gradient_eq}
holds. Namely, assuming the contrary, all the $t$-derivatives of the function~$\D(c(t), y)$
vanish. As the function~$\D(c(t), y)$ is real analytic in a neighborhood of~$t=0$
(as the composition of analytic functions is analytic), it follows that this function is
locally constant. This contradicts the fact that~$\D(c(t), y)$ changes sign at~$t=0$.

Assume conversely that for $x \in \overset{\circ}{\supp \rho}$  there is a point~$y\in \supp\rho\cap\K(x)$.
We choose the chain~$x_0, x_1=x, x_2, \ldots, x_k$ as in~\eqref{chain} with~$k=p+1$.
We use the notation of Corollary~\ref{matrix} in case~$N=k$. Then the Gram matrix~$L$
becomes
\[ L=\big(f(\varepsilon|i-j|) \big)_{i,j=0,...,k}=\bpm 1 & f(\varepsilon) & \cdots & f(k \varepsilon) \\
f(\varepsilon) & 1 & & \vdots \\ \vdots & &\ddots  & \\ f(k \varepsilon) &  \cdots & &1 \epm . \]
Using the expansion~\eqref{expand_D} for~$K=k-1$, we obtain
\beq \label{Lexp} \begin{split}
L & =E+a_1 \varepsilon^2 \big(|i-j|^2 \big) +a_2 \varepsilon^4 \big( |i-j|^4 \big) \\
& \qquad \qquad
+\ldots+a_{k-1} \varepsilon^{2(k-1)} \big( |i-j|^{2(k-1)} \big) +\Oo\big(\varepsilon^{2k}\big) \:,
\end{split}
\eeq
where~$E$ denotes the matrix where all the matrix entries (also the off-diagonal entries) are equal to one,
and~$(|i-j|^q)$ is the matrix whose element~$(i,j)$ has the value~$|i-j|^q$.

Let us construct a vector~$v \in  \C^{k+1}$ such that the expectation value~$\langle v,Lv\rangle$
is of the order~$\Oo(\varepsilon^{2k})$. To this end, we take for~$v=(v_i)_{i=0}^{k}\in \C^{k+1}$ a 
non-trivial solution of the~$k$ linear equations
\beq \sum_{i=0}^{k} v_i =0,\quad \sum_{i=0}^{k}  i v_i =0,\quad \sum_{i=0}^{k}  i^2 v_i =0,\quad \ldots, \sum_{i=0}^{k}  i^{k-1}  v_i =0\:. \label{star}
\eeq
Then~$\langle v,E v\rangle=0$ and for all $l\in\{1,\ldots,k-1\}$,
\begin{align*}
 \langle v,(|i-j|^{2l}) v\rangle=& \sum_{i,j=0}^{k} v_i v_j |i-j|^{2l}=\sum_{i,j=0}^{k} v_i v_j \sum_{\nu=0}^{2l} \binom{2l}{\nu} i^\nu j^{2l-\nu}=\\
=& \sum_{i,j=1}^{k} v_i v_j \left(i^{2l}+2l \;i^{2l-1}j+\ldots+\binom{2l}{l}i^l j^l+\ldots+j^{2l}\right) .
\end{align*}
Each summand involves a power of~$i$ and a power of~$j$, where always one of these powers is
smaller than~$k$. Thus all summands vanish according to~\eqref{star},
\[ \langle v,(|i-j|^{2l}) v\rangle = 0 \:. \]
We conclude that in the expectation value~$\langle v,Lv\rangle$ with~$L$ according to~\eqref{Lexp}
all the terms except for the error term vanish,
\[ \langle v,Lv\rangle = \Oo(\varepsilon^{2k})\:. \]
Moreover, the solution $v$ can always be normalized by $v_0=1$, because setting~$v_0$ to zero, the system of
equations~\eqref{star}
can be rewritten with the square Vandermonde matrix, which has a trivial kernel. 

We next consider the setting of Corollary~\ref{matrix}, but now with~$N=k+1$ and~$x_{k+1}=y$.
Using~\eqref{gradient_eq} together with the fact that the points~$y$ and~$x_0$ are timelike
separated, we find that
\beq \label{bdef}
\L(x_0, y) = b \,\varepsilon^p + \Oo(\varepsilon^{p+1})
\eeq
for~$b \neq 0$. We choose the vector~$u \in \C^{k+2}$ as~$u = (\alpha v_0, \ldots, \alpha v_k, \beta)$
with~$\alpha, \beta \in \R$ (and $v=(v_0, \ldots, v_k)$ as above). Then
\begin{align*}
\langle u, L u \rangle_{\C^{k+2}} =
\bigg\langle \bpm \alpha \\ \beta \epm,
\bpm \Oo(\varepsilon^{2k}) & b \varepsilon^p +\Oo(\varepsilon^{p+1}) \\
b \varepsilon^p +\Oo(\varepsilon^{p+1}) & \L(y,y) \epm \bpm \alpha \\ \beta \epm \bigg\rangle_{\C^2}\:,
\end{align*}
where we combined~\eqref{bdef} with our normalization~$v_0=1$, and used 
that~$y$ is not timelike separated from~$x_1, \ldots, x_k$.
For sufficiently small~$\varepsilon$, the matrix in this equation has a negative determinant,
in contradiction to Corollary~\ref{matrix}.
\QED
This completes the proof of Theorem~\ref{thmmain2}.

\section{The Variational Principles on the Circle} \label{sec5}
As a simple starting point for a more detailed analysis, we now
consider the variational principles on the circle (see Example~(b) on page~\pageref{Bcircle}).
We first discuss numerical results, which again show the ``critical behavior''
discussed in Section~\ref{sec3} for the variational principle on~$S^2$.
Applying the previous structural results, we will prove this critical behavior
and show, under generic assumptions, that the minimizing measure is supported at a finite number
of points. Moreover, we will give many minimizers in closed form.

Our numerical methods and results are similar to those on~$S^2$, as we now describe.
We again consider the weighted counting measure~\eqref{weightcount}.
As the starting configuration we choose,
in analogy of the Tammes distribution on~$S^2$
a uniform distribution of~$m$ points on the circle,
\beq \label{Xmdef}
X_m=\{x_k=e^{i(k-1) \vartheta_m},k=1,\ldots,m\} \:,\qquad
\vartheta_m=\frac{2\pi}{m} \:,
\eeq
with uniform weights~$\rho_k=1/m$. Minimizing as in
Section~\ref{sec3} with a simulated annealing algorithm,
we obtain the results shown in Figure~\ref{vgl_S1}.
\begin{figure}
 \centering
 \includegraphics[width=8cm]{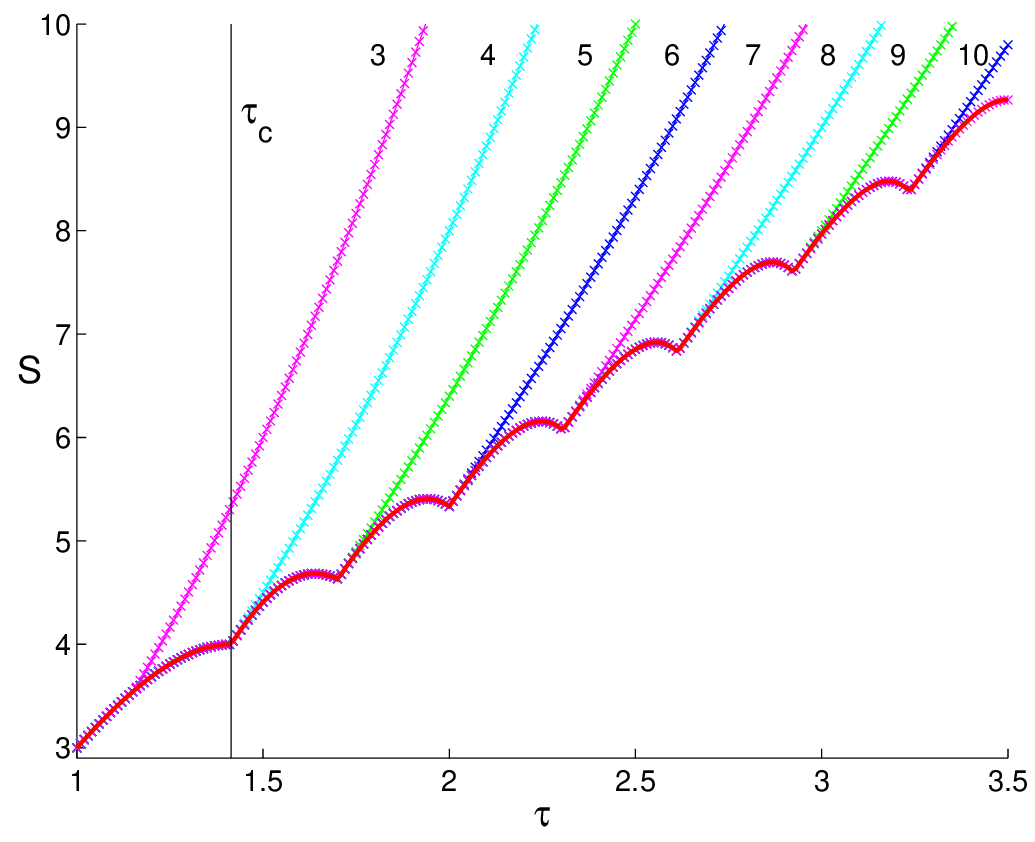}
\caption{Numerical minima for the weighted counting measure on the circle.}
\label{vgl_S1}
 \end{figure}
The numerical findings indicate that the minimizing measure is supported at a finite number of points $m_0$.
This number can be stated explicitly by
 \begin{equation}\label{m0}
m_0=\min \left\{ n\in \N\;:\; n\geq   \frac{2\pi}{\vartheta_{\max}}\right\}\:,                                                 \end{equation} 
where $\vtmax$, as given by \eqref{thetamax}, denotes the opening angle of the lightcone.
The number $m_0$ increases with~$\tau$, with discontinuous ``jumps'' at the values
\begin{equation} \label{taumS1}
\tau_m := \sqrt{\frac{2}{1-\cos(\vartheta_m)}}.
\end{equation}
Besides the discrete nature of the minimizers, the numerical results reveal that
at~$\tau=\tau_c = \sqrt{2}$ (corresponding to $\vtmax =\frac{\pi}{2}$), the structure of
the minimizers changes completely.
Just as in Section~\ref{sec3}, this effect can be understood as a phase transition.
More precisely, if~$\tau \leq \tau_c$, every minimizer is generically timelike.
If we further decrease~$\tau$ (i.e.,\ for every fixed~$\tau < \tau_3$), 
we even found a large number of minimizing measures, supported at
different numbers of points with strikingly different positions.
However, if~$\tau > \sqrt{2}$, the minimizer is unique (up to rotations on~$S^1$),
is supported at~$m_0$ points, and is not generically timelike.

In the remainder of this section, we make this picture rigorous.
First, the operator~$\D_\mu$ can be diagonalized explicitly by plane
waves~$\phi_n(x) = e^{i n \vartheta_x}$ (where~$n \in \Z$, and~$\vartheta_x$ is the
angle). This gives rise to the decomposition
\[ \D(x,y)=\nu_0 +\sum_{n=1}^2 \nu_n \left(e^{in (\vartheta_x - \vartheta_y)}+e^{-in (\vartheta_x - \vartheta_y)}\right),  \]
where
\beq \label{nu0form}
\nu_0=\iint_{S^1\times S^1} \D(x,y) \:d \mu(x)\, d \mu(y)=4 \tau^2-\tau^4
\eeq
and similarly~$\nu_1=2\tau^2$ and~$\nu_2=\frac{1}{2}\tau^4$.
In the case $\tau \leq 2$, all eigenvalues~$\nu_0$, $\nu_1$ and~$\nu_2$ are non-negative,
and we can apply Proposition~\ref{est_nu} to obtain
\[ \Sact_{\min} \geq \nu_0 \:. \]

For sufficiently small~$\tau$, the uniform distribution of points on the circle~\eqref{Xmdef}
gives a family of generically timelike minimizers:
\begin{Lemma}
If~$m\geq 3$ and~$\tau$ is so small that~$\L(x,y)=\D(x,y)$ for all~$x,y\in X_m$,
then $\rho=\frac{1}{m}\sum_{i=1}^m \delta_{x_i}$ is a generically timelike minimizer.
Every other minimizer is also generically timelike.
\end{Lemma}
\Proof
A straightforward calculation using the identities
\[ \sum\limits_{k=0}^{m-1} e^{ i k \vartheta_m}=0 \qquad \text{and} \qquad
\sum\limits_{k=0}^{m-1} \left( e^{i k\vartheta_m}\right)^2=0 \]
yields for any~$x \in S^1$,
\begin{align*}
 \dd(x)&=\frac{1}{m}\: 2 \tau^2 \sum_{k=0}^{m-1}
 \Big( 2+2 \, \langle x,x_k\rangle -\tau^2+\tau^2 \,\langle x,x_k\rangle^2 \Big) \\
&= \frac{1}{m}\: 2 \tau^2 \left(2m-m \tau^2+ \frac{m}{2}\: \tau^2 \right)=\nu_0 \:.
\end{align*}
In particular, one sees that~$\Sact[\rho]=\nu_0$.

The assumption~$\L(x,y)=\D(x,y)$ for all~$x,y\in X_m$ can be satisfied only if~$\tau < 2$.
Thus in view of~\eqref{nu0form}, the operator~$\D_\mu$ is positive.
We finally apply Proposition~\ref{est_nu}.
\QED
Applying this lemma in the case~$m=4$ gives the following result.
\begin{Corollary} \label{corS1}
If~$\tau \leq \tau_c$, every minimizer is generically timelike.
\end{Corollary} \noindent
More general classes of generically timelike minimizers can be constructed explicitly
with the help of Corollary~\ref{cor412}. In particular, one can find minimizing measures
which are not discrete. For the details we refer to the analogous measure
on~$S^2$ given in Example~\ref{ex61}.

Having explored the case $\tau \leq \tau_c$, we proceed with the case~$\tau > \tau_c$. As already stated, the closed lightcones are given by
\[ \J(x)=\Big\{y\;:\; \langle x,y\rangle  \geq 1-\frac{2}{\tau^2}=\cos(\vtmax) \Big\} \cup \{-x \} \:. \]
Therefore, if $\tau>\sqrt{2}=\tau_c$ (or equivalently $\vtmax<\frac{\pi}{2}$),
the condition of antipodal points (see Proposition~\ref{Prpgt}) is satisfied.
Thus there are no generically timelike minimizers. As the condition~\eqref{gradient}
is obvious, we can apply Theorem~\ref{thmmain1}~(A) and conclude that 
\beq \label{rhodiscrete}
\text{if~$\tau > \tau_c$, every minimizing measure is discrete}\:.
\eeq
Using results and methods from Section~\ref{secsingular}, we 
we will be able to explicitly construct all minimizers under the additional technical assumption that
\[ \tau > \tau_d := \sqrt{3+\sqrt{10}} \:. \]
We first introduce a descriptive notation:
\begin{Def}
 A \textbf{chain} of length $k$ is a sequence~$x_1,\ldots,x_k\in S^1$
 of pairwise distinct points such that $\langle x_i,x_{i+1}\rangle = \cos(\vartheta_{\max})$ for all $i=1,\ldots,k-1$.
\end{Def}

\begin{Thm} \label{thmS1}
If $\tau> \tau_d$, then the support of every minimizer~$\rho$ is a chain~$\{x_1, \ldots, x_{m_0} \}$
(with~$m_0$ as given by~\eqref{m0}).
The minimal action is 
\beq \label{Svalue}
\Smin= \frac{\Ll(0)(\Ll(0)+\Ll(\gamma))}{(m_0-2)(\Ll(0)+\Ll(\gamma))+2\Ll(0)}\;,
\eeq
where $\gamma=\arccos(\langle x_1,x_{m_0}\rangle)\in (0,\vartheta_{\max}]$.
The minimizing measure is unique up to rotations on~$S^1$.
\end{Thm} \noindent
An example for the support of the minimizing measure is shown in Figure~\ref{figchain_S1}.
\begin{figure}\label{figchain_S1}
\centering
\includegraphics[width=4cm]{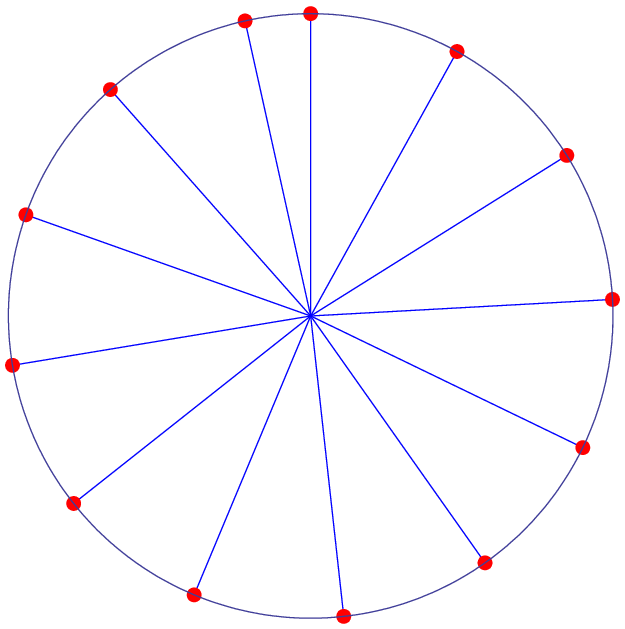}
\caption{A minimizer for $\tau=4$.}
\end{figure} %
Up to rotations, the points of the chain can be written as
\begin{equation}\label{x_chain}
x_k=e^{i (k-1) \vtmax },\quad k=1,\ldots,m_0.
\end{equation} 
In the special cases $\tau=\tau_m$, the minimizer is the measure with equal weights supported on the uniform distribution $X_{m}$. In the general case, the weights will not all be the same, as will be specified below.

For the proof of Theorem~\ref{thmS1} we proceed in several steps.
\begin{Lemma}
If $\tau>\sqrt{6}$, the minimal action is attained for a measure supported on a
chain~$x_1, \ldots, x_k$. In the case~$k=m_0$, every minimizing measure is
a chain.
\end{Lemma}
\Proof Let~$\rho$ be a minimizing measure.
We first note that every chain~$K$ in the support of~$\rho$ must have
finite length, because otherwise~$\vtmax/\pi$ would have to be irrational.
As a consequence, $K$ would be a dense set of~$S^1$,
in contradiction to the discreteness of~$\rho$ (see~\eqref{rhodiscrete}).
Let us assume that the support of~$\rho$ is not a chain.

We let~$K \subset \supp \rho$ be a chain, which is maximal in the sense that it cannot be extended.
Set~$L= \supp \rho \setminus K$. We consider variations of~$\rho$ where we rotate~$K$ by a small
angle~$\vartheta$, leaving the weights on~$K$ as well as $\rho|_L$ unchanged.
The fact that~$K$ cannot be extended implies that that these variations are smooth in~$\vartheta$
at~$\vartheta=0$. The minimality of~$\rho$ implies that
\beq \label{Svar}
\delta \Ss = 0 \qquad \text{and} \qquad
\delta^2 \Ss=\sum_{x\in K,y \in L} 2 \,\rho(x)\, \rho(y)\, \delta^2\Ll(x,y) \geq 0 
\:.
\eeq
On the other hand, differentiating~\eqref{DS2}, one finds that the function~$\D$
restricted to~$[0,\vartheta_{\max}]$ is concave,
\[ \D''(\vartheta)=-4\tau^2(\cos(\vartheta)+\tau^2 \cos(2\vartheta))<0 \qquad
(\text{if~$\tau > \sqrt{6}$})\:. \]
Comparing with~\eqref{Svar}, we conclude that~$\Ll(x,y)$ vanishes for all~$x \in K$
and~$y \in L$. In the case that~$\# K=m_0$, this implies that~$L=\varnothing$, a contradiction.
In the remaining case~$\# K < m_0$, we can subdivide the circle into two
disjoint arcs~$A_K$ and~$A_L$ such that~$K \subset A_K$ and~$L \subset A_L$.
The opening angle of~$A_K$ can be chosen larger than~$\vtmax$ times the length of~$K$,
giving an a priori upper bound on the length of~$K$.

By further rotating~$K$, we can arrange that the chain~$K$ can be extended by a point in~$L$,
without changing the action. If the extended chain equals the support of~$\rho$, the proof is 
finished. Otherwise, we repeat the above argument with~$K$ replaced by
its extension. In view of our a priori bound on the length of~$K$, this process ends after a finite
number of steps.
\QED

\begin{Lemma}
Suppose that~$\rho$ is a minimizing measure supported on a chain.
If $\tau>\sqrt{3+\sqrt{10}}$, the length of this chain is at most~$m_0$.
\end{Lemma}
\begin{proof}
For all $\gamma\in (0, \vartheta_{\max})$ an elementary calculation shows that
\begin{equation}\label{ineq_L}
 \Ll(\gamma)^2+\Ll(\vartheta_{\max}-\gamma)^2> \Ll(0)^2.
\end{equation} 
In the case~$\tau = \tau_{m_0}$ there is nothing to prove.
Thus we can assume that~$\tau\neq \tau_{m_0}$. For a chain~$x_1, \ldots, x_k$ with~$k>m_0$,
the Gram matrix corresponding to the points $x_1,x_{m_0+1},x_2$ has the form
\beq \label{Gramm}
\bpm \L(0) & \L(\vartheta_{\max}-\gamma) & 0 \\ \L(\vartheta_{\max}-\gamma) & \L(0) & \L(\gamma) \\ 0 & \L(\gamma) & \L(0) \epm .
\eeq
Using~\eqref{ineq_L}, its determinant is negative, in contradiction to Corollary~\ref{matrix}.
\end{proof}

From the last two lemmas we conclude that every minimizer~$\rho$
is supported on one chain of length at most~$m_0$.
If we parametrize the points as in~\eqref{x_chain}, the only contributions to the action come
from~$\Ll(x_l, x_l)$ and~$\Ll(x_1,x_{m_0})$. Using Lagrange multipliers,
the optimal weights $\rho_i=\rho(x_i)$ are calculated to be
\[ \rho_1=\rho_{m_0}=\frac{\lambda}{\Ll(0)+\Ll(\gamma)}\qquad \text{and}
 \qquad \rho_i=\frac{\lambda}{\Ll(0)}
\quad \text{for~$i=2,\ldots,m_0-1$}\:, \]
where we set
\begin{equation*}
 \lambda=\frac{\Ll(0)\: (\Ll(0)+\Ll(\gamma))}{(m_0-2)(\Ll(0)+\Ll(\gamma))+2\Ll(0)}\:.
\end{equation*} 
The corresponding action is computed to be~$ \Ss[\rho]=\lambda$, giving the formula 
in~\eqref{Svalue}. Using this explicit
value of the action, we obtain the following result.

\begin{Lemma} \label{lemmaS11}
Suppose that~$\rho$ is a minimizing measure supported on a chain.
Then the length of this chain is at least~$m_0$.
\end{Lemma}
\begin{proof}
For a chain of length~$n<m_0$, the only contributions to the action come
from~$\Ll(x_l, x_l)$, $l=1, \ldots, n$. The corresponding optimal weights are computed
by~$\rho_i = 1/n$. The resulting action is
\[ \Ss = \sum_{i=1}^n \frac{1}{n^2} \Ll(x_i,x_i) = \frac{1}{n}\Ll(0)\:. \]
This is easily verified to be strictly larger than the value of the action in~\eqref{Svalue}.
\end{proof} \noindent
This completes the proof of Theorem~\ref{thmS1}.

We finally remark that if~$\tau$ lies in the interval~$(\sqrt{2}, \sqrt{3+\sqrt{10}})$ where
Theorem~\ref{thmS1} does not apply, the numerics show that the minimizing $\rho$ is again the measure supported on the chain of length $m_0$, with one exception:
If~$\tau$ is in the interval~$(1.61988, \tau_5)$ with~$\tau_5=\sqrt{2+\frac{2}{\sqrt{5}}}$,
a chain of length $m_0+1=6$ gives a lower action than the chain of length $5$. In this case,
the Gram matrix~\eqref{Gramm} is indeed positive definite, so that the argument in Lemma~\ref{lemmaS11}
fails.

\section{The Variational Principles on the Sphere} \label{sec6}
We now come to the analysis of the variational principles on the sphere (see Example~(a) on page 4).
Applying Theorem~\ref{thmmain1}~(A) with the curve~$c$ chosen as the grand circle joining~$x$ and~$y$,
we immediately obtain that every minimizing measure~$\rho$ on~$S^2$ is either generically timelike
or else~$\overset{\circ}{\supp \rho} = \varnothing$.
The numerics in Section~\ref{sec3} indicated that these two cases are separated by
a ``phase transition'' at~$\tau=\tau_c = \sqrt{2}$. We will now prove that this phase transition
really occurs. Moreover, we will develop methods for estimating 
the minimal action from above and below. Many of these methods apply just as well to
the general setting introduced in Section~\ref{secmath} (see~\eqref{Ddef}--\eqref{var}).

\subsection{Generically Timelike Minimizers}
We first decompose~$\D$ in spherical harmonics. A short calculation yields
in analogy to~\eqref{Dspec} the decomposition
\[ \D(x,y)=\nu_0+4\pi  \sum_{l=1}^{2}\nu_l \sum_{m=-l}^l Y_l^m(x)\, \overline{Y_l^m(y)} \:, \]
where the eigenvalues are given by
\beq \label{nu0}
\nu_0= 4\,\tau^2-\frac{4}{3}\:\tau^4 \:,\qquad\nu_1=\frac{4}{3}\:\tau^2 \:,\qquad
\nu_2=\frac{4}{15}\:\tau^4\:.
\eeq
In particular, the operator $\D_\mu$ is positive if $\tau\leq\sqrt{3}$.

If~$\tau \leq \tau_c$, there is a large family of minimizers, as we now discuss.
The simplest example is the {\em{octahedron}}: 
Denoting the unit vectors in $\R^3$ by~$e_1,e_2,e_3$, we
consider the measure $\rho$ supported at $\pm e_i$ with equal weights $\tfrac{1}{6}$.
Obviously, the condition~(i) in Definition~\ref{defgtl} is satisfied. Moreover,
for any $x\in S^2$ one calculates 
\begin{align*}
 \dd(x)&=\frac{1}{6}\sum_{y\in \supp\rho} 2\tau^2\left(2+2 \langle x,y\rangle -\tau^2+\tau^2 \langle x,y\rangle^2\right)=\\
&=\frac{1}{3}\tau^2 \left(12 -6\tau^2+2\tau^2( x_1^2+x_2^2+x_3^2)\right)=\nu_0 \:.
\end{align*}
Thus Proposition~\ref{est_nu} yields that~$\rho$ is a
generically timelike minimizer. Moreover, from Proposition~\ref{est_nu} we conclude that
every minimizer is generically timelike.
If conversely~$\tau> \tau_c$, the condition of antipodal points is fulfilled,
and thus Proposition~\ref{Prpgt} shows that no generically timelike minimizers exist.
We have thus proved the following result:
\begin{Corollary} \label{cor61}
If~$\tau \leq \tau_c$, every minimizing measure~$\rho$ on~$S^2$ is generically timelike,
and the minimal action is equal to~$\nu_0$ as given by~\eqref{nu0}.
If conversely~$\tau > \tau_c$, every minimizing measure~$\rho$ is not generically timelike
and~$\overset{\circ}{\supp \rho} = \varnothing$.
\end{Corollary}
%
%Similarly, one can show that if~$\tau \leq \sqrt{\frac{3}{2}}$, the equal weighted measure supported at the tetrahedron is another example for a generically timelike minimizer.

Using Corollary~\ref{cor412}, one can also construct minimizers which are not discrete,
as is illustrated by the following example.
\begin{Example} \label{ex61}
{\em{ We introduce the function~$f\in L^2(S^2)$ by
 \[ f(\vartheta,\varphi)=
\begin{cases}
 \frac{5}{3} & \text{if } \vartheta \in [0,\arccos(0.8)] , \\[0.1em]
 \frac{35}{9} & \text{if } \vartheta \in [\arccos(0.4),\arccos(0.2)] \\[0.1em]
 \frac{40}{9} & \text{if } \vartheta \in [\arccos(-0.5),\arccos(-0.7)],\\[0.1em]
 0 & \text{otherwise.}
\end{cases} \]
Then if $\tau<1.00157$, a straightforward calculation shows that~$f$
has the properties (a) and (b) of Corollary~\ref{cor412}.
Thus the measure $d \rho=f d \mu$ is a minimizing generically timelike measure with $\overset{\circ}{\supp \rho}\neq \varnothing$. }} \QEDrem
\end{Example}

\subsection{Estimates of the Action}
As not even the solution of the Tammes problem is explicitly known, we cannot expect to
find explicit minimizers for general~$\tau$. Therefore, we need good estimates of the action
from above and below. We now explain different methods for getting estimates, which are
all shown in Figure~\ref{figestimates}.
\begin{figure}
 \centering
 \includegraphics[width=12cm]{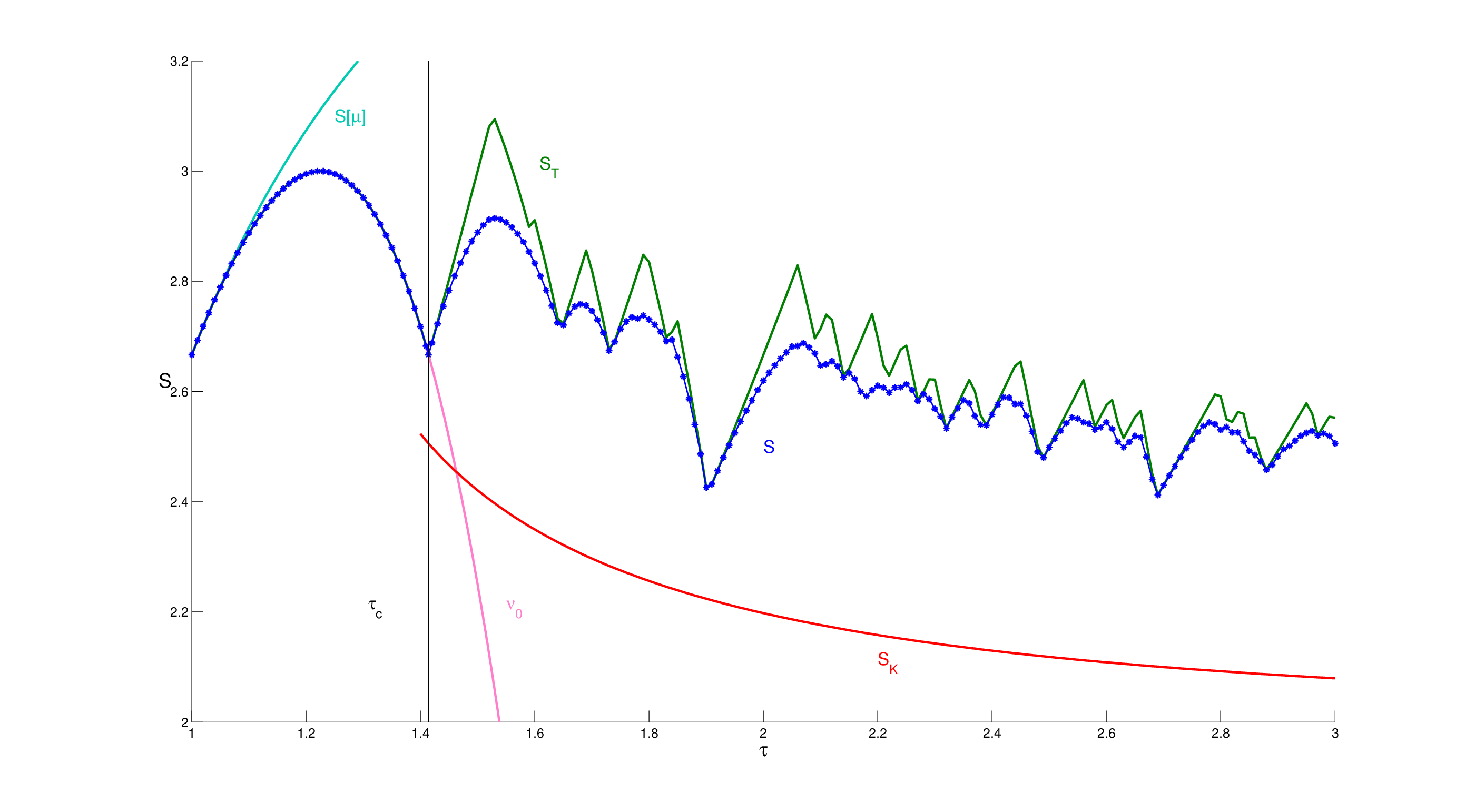}
\caption{Estimates of the action on~$S^2$: Upper bounds obtained from the volume measure~$S[\mu]$
and from the Tammes distribution~$S_T$, lower bounds by~$\nu_0$ and by the heat kernel
estimate~$S_K$.}
\label{figestimates}
 \end{figure} 

Estimates from above can be obtained simply by computing the action for suitable test measures.
For example, the action of the normalized volume measure is
\[ \Ss[\mu]= \frac{1}{4\pi} \int_0^{2\pi} d \varphi \;\int_0^{\vartheta_{\max}} d\vartheta\;\sin\vartheta\; \D (\vartheta)=4-\frac{4}{3\tau^2}\geq \Smin\: . \]
As one sees in Figure~\ref{figestimates}, this estimate is good if~$\tau$ is close to one.
Another example is to take the measure supported at the Tammes distribution for~$K$ points,
with equal weights. We denote the corresponding action by~$\Ss_T^K$. We then obtain the estimate
\[ \Smin \leq \Ss_T := \min_{K} \Ss_T^K\:. \] 
One method is to compute~$\Ss_T$ numerically using the tables in~\cite{sloane}.
This gives quite good results (see Figure~\ref{figestimates}), with the obvious disadvantage that
the estimate is not given in closed form. Moreover, the Tammes distribution is useful for
analyzing the asymptotics for large~$\tau$. To this end, for every Tammes-distribution $X_K$
we introduce~$\tau_K$ as the minimal value of~$\tau$ for which all distinct points in $X_K$ are 
spacelike separated. In analogy to~\eqref{taumS1}, the value of~$\tau_K$ is  given by
\[ \tau_K=\sqrt{\frac{2}{1-\cos(\vartheta_K) }} \:, \]
where~$\vartheta_K$ now denotes the minimal angle between the points of the Tammes distribution,
$$\vartheta_K=\max_{x_1,\ldots,x_K\in S^2} \min_{i\neq j}\; \arccos(\langle x_i,x_j\rangle) \:. $$
Using an estimate by W.\ Habicht and B.L.\ van der Waerden for the solution $\vartheta_K$
(see~\cite[page 6]{saff+kuijlaars}), we obtain
\[ 4 \left(\bigg (\frac{8 \pi}{\sqrt{3} K}\right)^{1/2} -\frac{C}{K^{2/3}} \bigg)^{-2}
\;\geq\; \tau_K^2 \geq 4 \,\frac{\sqrt{3} K}{8 \pi} \]
for some constant~$C>0$.
For given $\tau>1$ we choose $K\in \N$ such that $\tau_{K-1}\leq \tau<\tau_K$. Then 
\begin{align*}
\Smin \leq \Ss_T^{K-1} =
\frac{8\tau^2}{K-1}<\frac{8\tau^2_{K}}{K-1}\leq 32 \:\frac{K}{K-1}\bigg( \bigg(\frac{8 \pi}{\sqrt{3} }\bigg)^{1/2} -\frac{C}{K^{1/6}}\bigg)^{-2} .
\end{align*}
In the limit~$\tau \rightarrow \infty$, we know that~$K \rightarrow \infty$, and thus
\[ \limsup_{\tau \rightarrow \infty} \Smin \leq \frac{4\sqrt{3}}{\pi}\:.  \]

Constructing a lower bound is more difficult. From~\eqref{nu0} it is obvious that
the operator $\D_\mu$ is positive if $\tau \leq \sqrt{3}$.
Thus we can apply Proposition~\ref{est_nu} to obtain
$$\Smin\geq\nu_0 \qquad \text{if~$\tau \leq \sqrt{3}$}\:. $$
If~$\tau \leq \sqrt{2}$, this lower bound is even equal to~$\Smin$ according to
Corollary~\ref{cor61}. As shown in Figure~\ref{figestimates}, the estimate
is no longer optimal if~$\tau > \sqrt{2}$.

Another method for obtaining lower bounds is based on the following observation:
\begin{Prp}\label{prp_k_gen}
Assume that~$K_\mu$ is an integral operator on~$\Hil_\mu$ with integral
kernel~$K\in C^0(S^2\times S^2,\R)$ with the following properties:
\begin{itemize}
\item[(a)] $K(x,y)\leq \L(x,y)$ for all $x,y\in S^2$. \\[-1em]
\item[(b)] The operator~$K_\mu$ is positive.
\end{itemize}
Then the minimal action satisfies the estimate
\[ \Smin\geq  \iint_{S^2\times S^2} K(x,y) \:d \mu(x) \,d \mu(y)\:. \]
\end{Prp}
\Proof For any~$\rho \in \M$, our assumption~(a) gives rise to the estimate
\[ \Ss[\rho] =\iint_{S^2\times S^2} \L(x,y) \,d \rho(x)\, d \rho(y)
\geq \iint_{S^2\times S^2} K(x,y) \,d \rho(x)\, d \rho(y) \:. \]
Next, using property~(b), we can apply Proposition~\ref{prp_hom_min} to conclude that
the volume measure~$\mu$ is a minimizer of the variational principle corresponding to~$K$, i.e.
\[ \iint_{S^2\times S^2} K(x,y)\: d\rho(x)\, d\rho(y) \geq 
\iint_{S^2\times S^2} K(x,y)\: d\mu(x)\, d\mu(y) \:. \]
Combining these inequalities gives the result.
\QED
In order to construct a suitable kernel, we first consider the heat kernel~$h_t$ on $S^2$,  
\[ h_t(x,y) = \left( e^{t \Delta_{S^2}} \right)(x,y) =
4\pi \sum_{l=0}^{\infty} e^{-t\,l(l+1)} \sum_{m=-l}^l Y_l^m(x) \,\overline{Y_l^m(y)}\:.  \]
The heat kernel has the advantage that condition~(b) is satisfied, but condition~(a) is violated.
This leads us to choosing~$K$ as the difference of two heat kernels,
\[ K(x,y)=\lambda \,\big( h_{t_1}(x,y)-\delta h_{t_2}(x,y) \big) \:. \]
For given~$t_1<t_2$, we choose~$\delta$ and~$\lambda$ such that~$K(x,x)=1$ and~$K(\vtmax)=0$, i.e.
\[ \delta=\frac{h_{t_1}(\vartheta_{\max})}{h_{t_2}(\vartheta_{\max})} < 1
\qquad \text{and} \qquad
\lambda=\frac{\L(0)}{h_{t_1}(0)-\delta \:h_{t_2}(0)} > 0 \:. \]
By direct inspection, one verifies that condition~(a) is satisfied (see Figure~\ref{plot_L_K}
for a typical example).
\begin{figure}
\centering
\includegraphics[width=7cm]{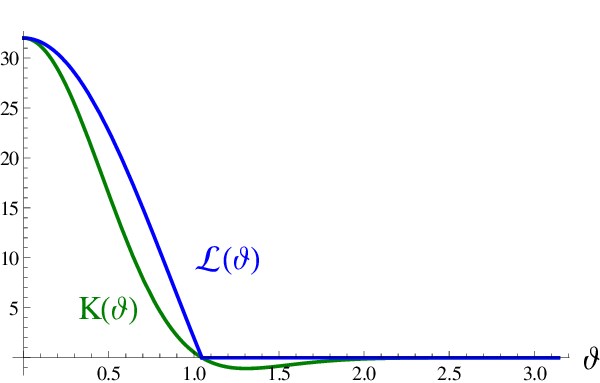}
\caption{The Lagrangian~$\L$ and the function~$K$ in the heat kernel estimate for~$\tau=2$.}
\label{plot_L_K}
\end{figure}
The eigenvalues of the operator~$K_\mu$ are computed to be
\[ \lambda\, (e^{-t_1\,l (l+1)}-\delta \,e^{-t_2\, l (l+1)}) \:, \]
showing that the operator $K_\mu$ is indeed positive.
Thus we can apply Proposition~\ref{prp_k_gen}. Using that
\[ \iint_{S^2 \times S^2} 
h_t(x,y)\: d\mu(x)\: d\mu(y) = 
\iint_{S^2 \times S^2} 4\pi \:Y_0^0(x) \overline{Y_0^0(y)}\: d\mu(x)\: d\mu(y) = 1 \:, \]
we obtain the {\bf{heat kernel estimate}}
$$ \Smin \geq S_K = \lambda\, (1-\delta) \:. $$
In this estimate, we are still free to choose the parameters~$t_1$ and~$t_2$.
By adjusting these parameters, one gets the lower bound shown in Figure~\ref{figestimates}.
Thus the heat kernel estimate differs from the minimal action only by an error of about $20 \%$,
and describes the qualitative dependence on~$\tau$ quite well. But of course, it does not
take into account the discreteness of the minimizers.

\section{The Variational Principles on the Flag Manifold~$\F^{1,2}(\C^f)$} \label{sec7}
We finally make a few comments on the variational principles on the flag manifold~$\F^{1,2}(\C^f)$
in the case~$f >2$ (see Example (c) on page~\pageref{Bcircle}). We first apply our main
Theorems~\ref{thmmain1} and~\ref{thmmain2} to obtain the following general result.

\begin{Thm}  \label{corfm} Every minimizer~$\rho$ on~$\F^{1,2}$ is either generically timelike
or~$\overset{\circ}{\supp \rho} = \varnothing$.
\end{Thm}
\Proof As a homogeneous space, the flag manifold~$\F^{1,2}(\C^f)$ has a real analytic structure
(see~\cite[Chapter II, \S4]{helgason2}). Then the function $\D$ is obviously real analytic. Moreover,
it is symmetric and constant on the diagonal. In order to apply Theorem~\ref{thmmain1},
for given~$y \in \K(x)$ we must find a curve~$c$ joining~$x$ and~$y$ which satisfies~\eqref{gradient}.
Alternatively, in order to apply Theorem~\ref{thmmain2}, our task is to construct a curve~$c(t)$
with~$c(0)=x$ and~$c(1)=y$ which is analytic in a neighborhood of~$t=0$, such that the
function~$\D(c(t), y)$ changes sign at~$t=0$.

We denote the range of~$x$ by~$I \subset \C^f$ and the orthogonal projection to~$I$ by~$\pi_I$.
Choosing an orthonormal basis~$(e_1, e_2)$ of~$I$, the matrix~$x|_I$ can be represented in terms of
Pauli matrices by
\[ x|_I = \1 + \tau \, \vec{u} \vec{\sigma} \qquad \text{with~$\vec{u} \in S^2$} \:. \]
Similarly, the operator~$\tilde{y} := \pi_I y \pi_I$ has the representation
\[ \tilde{y}|_I = \rho \1 + \kappa \, \vec{v} \vec{\sigma} \qquad \text{with~$\vec{v} \in S^2$}\:, \]
where the real parameters~$\rho$ and~$\kappa$ satisfy the inequalities
\[ 1-\tau \leq \rho-\kappa \leq 0 \leq \rho+\kappa \leq 1+\tau\:. \]
Using~\eqref{Dsimple}, the function~$\D$ is computed by
\[ \D(x,y) = 2 \,\Big( (\rho \tau + \kappa \cos \vartheta)^2 - \kappa^2 \, (\tau^2 -1)\: \sin^2 \vartheta \Big) , \]
where~$\vartheta$ denotes the angle between~$\vec{u}$ and~$\vec{v}$.
The operator~$\tilde{y}$ has rank two if and only if~$\kappa > |\rho|$.
A short calculation shows that in this case, $\D$ has only transverse zeros. Thus we can
choose a direction~$\dot{c}(0)$ where the condition~\eqref{gradient} is satisfied. Choosing a smooth
curve starting in this direction which joins~$x$ and~$y$, we can apply Theorem~\ref{thmmain1}~(A) to
conclude the proof in this case.

It remains to consider the situation when~$\tilde{y}$ has rank at most one. This leads us to
several cases. We begin with the case when~$y|_{I}$ vanishes. In this case,
we may restrict attention to
the four-dimensional subspace $U=\Im \,x \oplus \Im \,y$. In a suitable basis~$(e_1, \ldots, e_4)$ of
this subspace, the operators~$x$ and~$y$ have the matrix representations
\[ x=\bpm 1 & 0 \\ 0 & 0 \epm \otimes (\1+\tau \, \vec{u} \vec{\sigma} ) \:,\qquad
y=\bpm 0 & 0 \\ 0 & 1\epm \otimes (\1+\tau\; \vec{v} \vec{\sigma} ) \:, \]
where again~$\vec{u}, \vec{v} \in S^2$. A unitary transformation of the basis vectors~$e_1$
and~$e_2$ describes a rotation of the vector~$\vec{u}$ in~$\R^3$.
By a suitable transformation of this type, we can arrange that the angle between~$\vec{u}$ and~$\vec{v}$
equals~$\vtmax$ (see~\eqref{thetamax}).
We now define the curve $c:[0,\pi]\rightarrow \F^{1,2}$ by
\beq \label{cform}
c(t)=\bpm \cos(t)^2 & \sin(t) \cos(t) \\ \sin(t) \cos(t) & \sin(t)^2 \epm \otimes (\1+\tau \,\vec{w}(t)
\, \vec{\sigma}),
\eeq
where $\vec{w} :[0,\pi]\rightarrow S^2$ is the geodesic on $S^2$ with $\vec{w}(0)=\vec{u}$
and $\vec{w}(\pi)=\vec{v}$.
The curve $c$ is a real analytic function with $c(0)=x$ and $c(\pi)=y$, which is obviously
translation symmetric. Furthermore, one computes
\[ \D(c(t),y)=\sin(t)^4 \:\D_{S^2} (\vec{w}(t), \vec{v}), \]
where $\D_{S^2}$ is the corresponding function on the unit sphere~\eqref{DS2}.
As~$\D_{S^2}(\vartheta)$ changes sign at $\vtmax$, the function $\D(c(t),y)$ changes sign at $t=0$. Thus Theorem~\ref{thmmain2}~(A) applies, completing the proof in the case~$y|_{I}=0$.

We next consider the case that~$\tilde{y}$ has rank one. We choose the basis~$(e_1, e_2)$
of~$I$ such that~$\tilde{y}$ is diagonal,
\[ \tilde{y} = \bpm a & 0 \\ 0 & 0 \epm \qquad \text{with~$a \neq 0$} . \]
An elementary consideration shows that we can extend the basis of~$I$ to an orthonormal
system~$(e_1, e_2, e_3)$ such that on the subspace~$J := \langle \{e_1, e_2, e_3\} \rangle$,
the operator~$\hat{y} := \pi_J y \pi_J$ has the form
\beq \label{y3form}
\hat{y} |_{\langle \{e_1, e_2, e_3\} \rangle} = \bpm a & 0 & \overline{b} \\
0 & 0 & 0 \\
b & 0 & c \epm \qquad \text{with~$a \neq 0$ and~$ac \neq |b|^2$} \:.
\eeq
We let~$U$ be the unitary transformation
\[ U(t)|_J  = \bpm 1 & 0 & 0 \\
0 & \cos t & \sin t \\ 0 & -\sin t & \cos t
\epm  \quad \text{and} \quad
U(t)|_{J^\perp} = \1\:. \]
Setting~$y(t) = U(t) \,y\, U(t)^{-1}$, the matrix~$\tilde{y}$ becomes
\[ \tilde{y}(t) = \bpm 1 & 0 \\ 0 & \sin t \epm \left( \rho\,\1 + \kappa\, \vec{v} \vec{\sigma} \right)
\bpm 1 & 0 \\ 0 & \sin t \epm , \]
where~$\rho$ and~$\kappa$ are new parameters with
\beq \label{ineqs}
\kappa > |\rho| \qquad \text{and} \qquad \rho + \kappa \,v_3 = a \neq 0
\eeq
and~$\vec{v} \in S^2$ is again a unit vector.
The function~$\D$ is now computed by
\begin{align}
\D(x, y(t)) &= \frac{1}{2} \: \Tr \big( x|_I\, \tilde{y}(t) \big)^2 - 2 \det \big( x|_I\, \tilde{y}(t) \big) 
\nonumber \\
%=  \frac{1}{2} \: \Tr(x|_I\, \tilde{y}(t))^2 - 2 \det(x|_I) \, \det(\tilde{y}(t)) \\
&= \frac{1}{2} \Big(  \Tr \big( x|_I\, \tilde{y}(t) \big)^2 - 4\: (\tau^2-1)\: (\kappa^2 - \rho^2)\,\sin^2 t \Big) .
\label{Dform}
\end{align}
In order to simplify the trace, we transform the phase of~$e_3$. This changes the phase of~$b$
in~\eqref{y3form}, thus describing a rotation of
the vector~$\vec{v}$ in the $(1,2)$-plane. This makes it possible to arrange that
the vectors~$(u_1, u_2)$ and~$(v_1, v_2)$ are orthogonal in~$\R^2$. We thus obtain
\[ \Tr \big( x|_I\, \tilde{y}(t) \big) = 
(1+\tau u_3) (\rho + \kappa v_3) + (1-\tau u_3)(\rho - \kappa v_3)\: \sin^2 t\:. \]
We now have two subcases:
\begin{itemize}
\item[(1)] $v_3\neq \pm 1$: 
We vary the vectors~$\vec{u}$ and~$\vec{v}$ as functions of~$t$ such that
the above orthogonality relations remain valid and
\[ u_3 = \cos(\vartheta + \alpha t) \:, \qquad v_3 = \cos(\varphi + \beta t) \]
with free ``velocities'' $\alpha$ and~$\beta$. Since~$\L(x, y)=0$ at~$t=0$, we know that
\beq \label{sinth}
\cos \vartheta = -\frac{1}{\tau} \:,\qquad \sin \vartheta = \frac{\sqrt{\tau^2-1}}{\tau} \neq 0 \:.
\eeq
A Taylor expansion yields
\begin{align}
\qquad\quad  \Tr \big( &x|_I\, \tilde{y}(t) \big) = -t \:\alpha \tau \: (\rho + \kappa v_3) \: \sin \vartheta \label{Trlin} \\
&+ \frac{t^2}{2} \Big( (4 + \alpha^2)\: \rho + (-4+\alpha^2) \kappa \cos \varphi
+ 2 \alpha \beta \kappa \tau\: \sin \vartheta \: \sin \varphi \Big)
+ {\mathcal{O}}(t^3)\:. \label{Trquad}
\end{align}
As the factor~$(\rho + \kappa v_3)$ is non-zero in view of~\eqref{ineqs}, the linear term~\eqref{Trlin} does not vanish whenever~$\alpha \neq 0$.
By suitably adjusting~$\alpha$, we can arrange that
the square of this linear term compensates the last term in~\eqref{Dform}
(which is also non-zero in view of our assumption~$\kappa>\rho$).
Next, we know from~\eqref{sinth} and our assumptions
that the term~$\sim \alpha \beta$ in~\eqref{Trquad} is non-zero. Thus by a suitable
choice of~$\beta$, we can give the quadratic term~\eqref{Trquad} any value we want.
Taking the square, in~\eqref{Dform} we get a contribution~$\sim t^3$.
Thus the function~$\D$ changes sign.
Transforming to a suitable basis where~$y$ is a fixed matrix, we obtain a curve~$x(t)$
which is locally translation symmetric. Extending this curve to a smooth curve~$c$
which joins the point~$y$, we can apply Theorem~\ref{thmmain2}~(A).
\item[(2)] $v_3 =\pm 1$: We know that the matrix~$\tilde{y}$
is diagonal,
\[ \tilde{y}(t) = \bpm \rho \pm \kappa & 0 \\ 0 & (\rho \mp \kappa) \sin^2 t \epm \:. \]
Now we keep~$v$ fixed, while we choose the curve~$u(t)$ to be a great circle which is inclined
to the~$(1,3)$-plane by an angle~$\gamma \neq 0$, i.e.
\[ u_3 = \cos (\vartheta + \alpha t)\, \cos \gamma \:. \]
Repeating the above calculation leading to~\eqref{Trlin} and~\eqref{Trquad}, one sees
that we again get a non-zero contribution to~$\D$ of the order~$\sim t^3$.
Thus~$\D$ again changes sign, making it possible to apply Theorem~\ref{thmmain2}~(A).
\end{itemize}

It remains to consider the case when~$\tilde{y}$ vanishes but~$y|_I \neq 0$.
A short consideration shows that~$y|_I$ cannot have rank two. Thus we
can choose the orthonormal basis~$(e_1, e_2)$ of~$I$ such that~$y e_1 \neq 0$
and~$y e_2 = 0$. By suitably extending this orthonormal system by~$e_3$ and~$e_4$,
we can arrange that the operator~$y$ is invariant
on the subspace~$\langle \{e_1, e_2, e_3, e_4 \} \rangle$ and has the matrix representation
\[ y|_{\langle \{e_1, e_2, e_3, e_4\} \rangle} = \bpm 0 & 0 & \overline{a} & 0 \\
0 & 0 & 0 & 0 \\
a & 0 & c & \overline{b} \\
0 & 0 & b & 0 \epm . \]
If~$b \neq 0$, we can again work with the curve~\eqref{cform}.
If on the other hand~$b =0$, the operator~$y$ is invariant on~$\langle \{e_1, e_2, e_3\} \rangle$
and has the canonical form
\[ y|_{\langle \{e_1, e_2, e_3\} \rangle} = \bpm 0 & 0 & \sqrt{\tau^2-1} \\
0 & 0 & 0 \\
\sqrt{\tau^2-1} & 0 & 2 \epm . \]
Transforming~$y$ by the unitary matrix
\[ V(\tau) \bpm e_1 \\ e_3 \epm = \bpm \cos \tau & \sin \tau \\ -\sin \tau & \cos \tau
\epm \bpm e_1 \\ e_3 \epm \:, \]
we can arrange that~$y$ is again of the form~\eqref{y3form}, but now with
coefficients depending on~$\tau$. Setting~$t=\tau^2$, we can again use the construction
after~\eqref{y3form}. This completes the proof.
\QED

For sufficiently large~$\tau$, we can rule out one of the cases in Theorem~\ref{corfm},
showing that the minimizing measures do have a singular support.
\begin{Thm} \label{thmfm} There are no generically timelike minimizers if
\[ \tau^2 > \frac{3 f + 2 \sqrt{3\:(f^2-1)}}{(2 + f)} \:. \]
\end{Thm} \noindent
The method of proof is to apply Proposition~\ref{Prpgt}~(I). In the next two lemmas we
verify the necessary assumptions and compute~$\nu_0$.
\begin{Lemma}
The operator $\D_\mu$ has rank at most $3 f^4$.
\end{Lemma}
\Proof We extend the method used in the proof of~\cite[Lemma~1.10]{continuum}. 
A point~$x \in \F$ is a Hermitian $f \times f$-matrix of rank two, with non-trivial
eigenvalues~$1+\tau$ and~$1-\tau$. Thus we can represent~$x$ in in bra/ket notation as
\[ x = |u(x) \rangle \langle u(x) | - |v(x) \rangle \langle v(x)|\:, \]
where~$u(x)$ and~$v(x)$ are the eigenvectors of~$x$, normalized such that
\[ \langle u(x) |u(x) \rangle= \tau+1 \qquad \text{and} \qquad \langle v(x) |v(x) \rangle=\tau-1\:. \]

A short calculation shows that the non-trivial eigenvalues
of the matrix product~$xy$ coincide with the eigenvalues of the $2 \times 2$-matrix
product
\[ A_{xy} := \begin{pmatrix} \langle u(x) | u(y) \rangle & -\langle u(x) | v(y) \rangle \\
\langle v(x) | u(y) \rangle & -\langle v(x) | v(y) \rangle \end{pmatrix}
\begin{pmatrix} \langle u(y) | u(x) \rangle & -\langle u(y) | v(x) \rangle \\
\langle v(y) | u(x) \rangle & -\langle v(y) | v(x) \rangle \end{pmatrix} . \]
Using~\eqref{Dsimple}, we can thus write the function~$\D$ as
\[ \D(x,y) = \Tr \bigg[ \Big( A_{xy} - \frac{1}{2}\: \Tr (A_{xy}) \Big)^2 \bigg] . \]
This makes it possible to recover $\D(x,y)$ as the ``expectation value''
\[ \D(x,y) = \bigg\langle \begin{pmatrix} u \otimes u^* \otimes u \otimes u^* \\
u \otimes u^* \otimes v \otimes v^* \\
v \otimes v^* \otimes v \otimes v^* \end{pmatrix} \!\bigg|_x,
B
\begin{pmatrix} u \otimes u^* \otimes u \otimes u^* \\
u \otimes u^* \otimes v \otimes v^* \\
v \otimes v^* \otimes v \otimes v^* \end{pmatrix} \!\bigg|_y \bigg\rangle_{\C^{3 f^4}} \]
of a suitable matrix~$B$, whose $3 \times 3$ block entries are of the form
\[ B_{ij} = b_{ij} +  \delta_{i,2} \delta_{j,2}\: (c_1 \rho_1 + c_2 \rho_2 + c_3 \rho_3)
\quad \text{with} \quad b_{ij}, c_i \in \C, \]
and the operators~$\rho_i$ permute the factors of the tensor product,
\begin{align*}
\rho_1(u \otimes u^* \otimes v \otimes v^*) &= v \otimes v^* \otimes u \otimes u^* \\
\rho_2(u \otimes u^* \otimes v \otimes v^*) &= u \otimes v^* \otimes v \otimes u^* \\
\rho_3(u \otimes u^* \otimes v \otimes v^*) &= v \otimes u^* \otimes u \otimes v^* \:.
\end{align*}
Hence introducing the operator
\[ K \::\: L^2(\F, d\mu_L) \rightarrow \C^{3f^4} \::\: \psi \mapsto
\int_\F \begin{pmatrix} u \otimes u^* \otimes u \otimes u^* \\
u \otimes u^* \otimes v \otimes v^* \\
v \otimes v^* \otimes v \otimes v^* \end{pmatrix} \!\Bigg|_x\:
\psi(x)\:d\mu_L(x)\:, \]
we find that $\D_\mu = K^* B K$. This gives the claim.
\QED
In view of this lemma, we may decompose~$\D$ in the form~\eqref{Dspec}.
\begin{Lemma} The eigenvalue~$\nu_0$ in the decomposition~\eqref{Dspec} is given by
\[ \nu_0=\frac{2(3f+6f\tau^2-(2+f)\tau^4-6)}{f(f^2-1)} \:. \]
\end{Lemma}
\Proof It is most convenient to represent the elements in~$\F$ as
\beq \label{xrep}
(1+\tau)\: |u \rangle \langle u | + (1-\tau)\: |v \rangle \langle v |\:,
\eeq
where the vectors~$u, v \in \C^f$ are orthonormal. Then the normalized volume measure $\mu$
on $\F$ can be written as
\[  d \mu=\frac{1}{\vol(\F)} \:\delta \big(\Rec\langle u,v\rangle \big)
\:\delta \big(\Imc\langle u,v\rangle \big) \:\delta \big(\|u\|^2-1\big)\: \delta \big( \|v\|^2-1 \big)\: d u \: d v,  \]
where~$du$ and~$dv$ denote the Lebesgue measure on~$\C^f$, and~$\delta$ is the Dirac distribution.
The total volume is computed to be
\begin{align*}
\vol(\F)&=\iint_{\C^f\times \C^f} \delta(\Rec\langle u,v\rangle)\:
\delta(\Imc\langle u,v\rangle)\: \delta(\|u\|^2-1)\delta(\|v\|^2-1)\: d u \: d v \\
&=\frac{1}{4} \:\vol (S^{2f-1})\:\vol (S^{2f-3})\: .
\end{align*}
To simplify the calculations, we fix~$x$ and choose an eigenvector basis of~$x$. Then~$x=\mathrm{diag}((1+\tau),(1-\tau),0,\ldots,0)$, whereas~$y$ is again represented in the form~\eqref{xrep}.
Then the eigenvalues of the product $xy$ depend only on the vector
components~$u_1, u_2$ and~$v_1, v_2$. More precisely, using~\eqref{Dsimple}, we obtain
\begin{align*}
 \D(x,y) \:=\:& \frac{1}{2} \:\Big[ (1+\tau)^2\, |u_1|^2+(1-\tau^2) \,(|v_1|^2-|u_2|^2)-(1-\tau)^2\, |v_2|^2 \Big]^2\\
&+2\, (1-\tau^2)\:  \big|(1+\tau) \,u_1\, \overline{u_2}+(1-\tau)\, v_1\, \overline{v_2} \big|^2 =:f(u,v).
\end{align*}
Our task is to compute the integral $ \nu_0=\int_{\F}f(u,v) \,d \mu$. In the case~$f \geq 4$,
one uses the symmetries to reduce to a lower-dimensional integral,
\begin{align*}
\nu_0  &= c
\int_0^\infty d u_1 \int_0^\infty d u_2  \int_0^\infty d u_3 \int_0^\infty d v_1 \int_\C d v_2\int_\C d v_3 \int_0^\infty d v_4 \\
& \quad \times \delta \big( \|u\|^2-1 \big) \:\delta \big( \|v\|^2-1 \big) \:
 \delta\big( \Rec \langle u,v\rangle \big) \:\delta \big( \Imc \langle u,v\rangle \big)\:
  f(u,v) \:u_1 \,u_2\, u_3^{2f-5}\, v_1\, v_4^{2f-7} ,
 \end{align*}
where~$c$ is the constant
\[ c = \frac{1}{\vol(\Ff)}\: \vol(S^{2f-5})\: \vol(S^{2f-7})\: (2 \pi)^3\:. \] 
Carrying out all integrals gives the claim. The proof in the case~$f=3$ is similar.
\QED

The remaining question is whether generically
timelike minimizers exist for small~$\tau$. In the special case~$\tau=1$, the operator
$\D_\mu = \L_\mu$ is positive (see~\cite[Lemma~1.10]{continuum}), so that Proposition~\ref{prp_hom_min} 
or similarly Proposition~\ref{est_nu} yields that the standard volume measure is a generically
timelike minimizer. However, if~$\tau>1$, these propositions can no longer be used, because the
operator~$\D_\mu$ fails to be positive:
\begin{Lemma}
If~$\tau>1$, the operator $\D_\mu$ has negative eigenvalues.
\end{Lemma}
\Proof
Since~$\supp\mu=\F$, it suffices to find two points~$x_1, x_2 \in \F$
such that the corresponding Gram matrix $\D(x_i, x_j)$ is not positive semi-definite.
For given~$\varepsilon\in (0,1)$ we choose the four vectors
$$u_1=e_1\:,\quad v_1=e_2 \quad \text{and} \quad u_2=e_1\:,\quad \;v_2=\sqrt{\varepsilon}\:
e_2+ \sqrt{1-\varepsilon} \:e_3 $$
(where $e_i$ are the standard basis vectors of $\C^f$).
Taking the representation~\eqref{xrep}, we obtain two points~$x_1, x_2 \in \F$.
The corresponding Gram matrix is computed to be
\[ \bpm 8\tau^2 & \frac{1}{2}\left(-\varepsilon (\tau-1)^2+(\tau+1)^2\right)^2 \\
\frac{1}{2}\left(-\varepsilon (\tau-1)^2+(\tau+1)^2\right)^2 & 8\tau^2 \epm . \]
The determinant of this matrix is negative for small~$\varepsilon>0$.
\QED
In this situation, Proposition~\ref{lemmaiff} still gives some information on the possible support of
generically timelike minimizers. But it remains an open problem whether and under
which conditions generically timelike minimizers exist.

{\Thanks{{\em{Acknowledgments:}} We would like to thank the referees for valuable comments on the manuscript.}

%\bibliographystyle{amsplain}
%\bibliography{../felix}
\providecommand{\bysame}{\leavevmode\hbox to3em{\hrulefill}\thinspace}
\providecommand{\MR}{\relax\ifhmode\unskip\space\fi MR }
\providecommand{\MRhref}[2]{%
  \href{http://www.ams.org/mathscinet-getitem?mr=#1}{#2}
}
\providecommand{\href}[2]{#2}

\end{document}